\documentclass[12pt]{article}
\usepackage{amssymb}
\usepackage{amsmath}
\usepackage{amsfonts}
\usepackage{latexsym}
\usepackage{hyperref}
\usepackage{graphicx}
\usepackage{amsthm}
\usepackage{color}

\newtheorem{algorithm}{Algorithm}
\newtheorem{assumption}{Assumption}

\newtheorem{claim}{Claim}

\newtheorem{corollary}[claim]{Corollary}

\newtheorem{lemma}{Lemma}

\theoremstyle{definition}
\newtheorem{theorem}{Theorem}
\newtheorem{definition}{Definition}
\newtheorem{example}{Example}

\begin{document}

\title{\textbf{Spread of Misinformation in Social Networks}\thanks{We would like to thank Pablo Parrilo and Devavrat
Shah for useful comments and discussions.}}
\author{\textbf{Daron Acemoglu,}\thanks{%
Department of Economics, Massachusetts Institute of Technology} \textbf{%
Asuman Ozdaglar,}\thanks{%
Department of Electrical Engineering and Computer Science, Massachusetts
Institute of Technology} \textbf{and Ali ParandehGheibi}\thanks{%
Department of Electrical Engineering and Computer Science, Massachusetts
Institute of Technology}}

\markboth{LIDS Report 2812}{LIDS Report 2812}

\maketitle

\pagestyle{myheadings}

\begin{center}
\textbf{Abstract}
\end{center}

We provide a model to investigate the tension between information
aggregation and spread of misinformation in large societies
(conceptualized as networks of agents communicating with each
other). Each individual holds a belief represented by a scalar.
Individuals meet pairwise and exchange information, which is
modeled as both individuals adopting the average of their
pre-meeting beliefs. When all individuals engage in this type of
information exchange, the society will be able to effectively
aggregate the initial information held by all individuals. There
is also the possibility of misinformation, however, because some
of the individuals are \textquotedblleft
forceful,\textquotedblright\ meaning that they influence the
beliefs of (some) of the other individuals they meet, but do not
change their own opinion. The paper characterizes how the presence
of forceful agents interferes with information aggregation. Under
the assumption that even forceful agents obtain some information
(however infrequent) from some others (and additional weak
regularity conditions), we first show that beliefs in this class
of societies converge to a consensus among all individuals. This
consensus value is a random variable, however, and we characterize
its behavior. Our main results quantify the extent of
misinformation in the society by either providing bounds or exact
results (in some special cases) on how far the consensus value can
be from the benchmark without forceful agents (where there is
efficient information aggregation). The worst outcomes obtain when
there are several forceful agents and forceful agents themselves
update their beliefs only on the basis of information they obtain
from individuals most likely to have received their own
information previously.

\vskip 1pc


\vskip .5pc


\setcounter{page}{1}

\begin{center}
\bigskip
\end{center}


\setcounter{page}{1}

\section{Introduction}

Individuals form beliefs on various economic, political and social
variables (``state") based on information they receive from others,
including friends, neighbors and coworkers as well as local leaders,
news sources and political actors. A key tradeoff faced by any
society is whether this process of information exchange will lead to
the formation of more accurate beliefs or to certain systematic
biases and spread of misinformation. A famous idea going back to
Condorcet's Jury Theorem (now often emphasized in the context of
ideas related to \textquotedblleft wisdom of the
crowds\textquotedblright ) encapsulates the idea that exchange of
dispersed information will enable socially beneficial aggregation of
information. However, as several examples ranging from the effects
of the Swift Boat ads during the 2004 presidential campaign to the
beliefs in the Middle East that 9/11 was a US or Israeli conspiracy
illustrate, in practice social groups are often swayed by misleading
ads, media outlets, and political leaders, and hold on to incorrect
and inaccurate beliefs.

A central question for social science is to understand the conditions under
which exchange of information will lead to the spread of misinformation
instead of aggregation of dispersed information. In this paper, we take a
first step towards developing and analyzing a framework for providing
answers to this question. While the issue of misinformation can be studied
using Bayesian models, non-Bayesian models appear to provide a more natural
starting point.\footnote{%
In particular, misinformation can arise in a Bayesian model if an agent
(receiver) is unsure of the type of another agent (sender) providing her
with information and the sender happens to be of a type intending to mislead
the receiver. Nevertheless, this type of misinformation will be limited
since if the probability that the sender is of the misleading type is high,
the receiver will not change her beliefs much on the basis of the sender's
communication.} Our modeling strategy is therefore to use a non-Bayesian
model, which however is reminiscent of a Bayesian model in the absence of
\textquotedblleft forceful\textquotedblright\ agents (who are either trying
to mislead or influence others or are, for various rational or irrational
reasons, not interested in updating their opinions).

We consider a society envisaged as a social network of $n$ agents,
communicating and exchanging information. Specifically, each agent is
interested in learning some underlying state $\theta \in \mathbb{R}$ and
receives a signal $x_{i}(0)\in \mathbb{R}$ in the beginning. We assume that $%
\theta =1/n\sum_{i=1}^{n}x_{i}(0)$, so that information about the relevant
state is dispersed and this information can be easily aggregated if the
agents can communicate in a centralized or decentralized fashion.

Information exchange between agents takes place as follows: Each individual
is \textquotedblleft recognized\textquotedblright\ according to a Poisson
process in continuous time and conditional on this event, meets one of the
individuals in her \emph{social neighborhood} according to a pre-specified
stochastic process. We think of this stochastic process as representing an
underlying \emph{social network }(for example, friendships, information
networks, etc.). Following this meeting, there is a potential exchange of
information between the two individuals, affecting the beliefs of one or
both agents. We distinguish between two types of individuals: \textit{regular%
} or \textit{forceful}. When two regular agents meet, they update their
beliefs to be equal to the average of their pre-meeting beliefs. This
structure, tough non-Bayesian, has a simple and appealing interpretation and
ensures the convergence of beliefs to the underlying state $\theta $ when
the society consists only of regular agents.\footnote{%
The appealing interpretation is that this type of averaging would be optimal
if both agents had beliefs drawn from a normal distribution with mean equal
to the underlying state and equal precision. This interpretation is
discussed in detail in De Marzo, Vayanos, and Zwiebel \cite{demarzo} in a
related context.} In contrast, when an agent meets a forceful agent, this
may result in the forceful agent \textquotedblleft
influencing\textquotedblright\ his beliefs so that this individual inherits
the forceful agent's belief except for an $\epsilon $ weight on his
pre-meeting belief.\footnote{%
When $\epsilon =1/2$, then the individual treats the forceful agent just as
any other regular agent (is not influenced by him over and above the
information exchange) and the only difference from the interaction between
two regular agents is that the forceful agent himself does not update his
beliefs. All of our analysis is conducted for arbitrary $\epsilon $, so
whether forceful agents are also \textquotedblleft
influential\textquotedblright\ in pairwise meetings is not important for any
of our findings.} Our modeling of forceful agents is sufficiently general to
nest both individuals (or media outlets) that purposefully wish to influence
others with their opinion or individuals who, for various reasons, may have
more influence with some subset of the population.\footnote{%
What we do not allow are individuals who know the underlying state and try
to convince others of some systematic bias relative to the underlying state,
though the model could be modified to fit this possibility as well.} A key
assumption of our analysis is that even forceful agents engage in some
updating of their beliefs (even if infrequently) as a result of exchange of
information with their own social neighborhoods. This assumption captures
the intuitive notion that \textquotedblleft no man is an
island\textquotedblright\ and thus receives some nontrivial input from the
social context in which he or she is situated.\footnote{%
When there are several forceful agents and none of them ever change
their opinion, then it is straightforward to see that opinions in
this society will never settle into a \textquotedblleft
stationary\textquotedblright\ distribution. While this case is also
interesting to study, it is significantly more difficult to analyze
and requires a different mathematical approach.} The influence
pattern of social agents superimposed over the social network can be
described by directed links, referred to as \textit{forceful links},
and creates a richer stochastic process, representing the evolution
of beliefs in the society. Both with and without forceful agents,
the evolution of beliefs can be represented by a Markov chain and
our analysis will exploit this connection. We will frequently
distinguish the Markov chain representing the evolution of beliefs
and the Markov chain induced by the underlying social network (i.e.,
just corresponding to the communication structure in the society,
without taking into account the influence pattern) and properties of
both will play a central role in our results.

Our objective is to characterize the evolution of beliefs and quantify the
effect of forceful agents on public opinion in the context of this model.
Our first result is that, despite the presence of forceful agents, the
opinion of all agents in this social network converges to a common, tough
stochastic, value under weak regularity conditions. More formally, each
agent's opinion converges to a value given by $\pi ^{\prime }x(0)$, where $%
x(0)$ is the vector of initial beliefs and $\pi $ is a random vector. Our
measure of spread of misinformation in the society will be $\bar{\pi}%
^{\prime }x(0)-\theta =\sum_{i=1}^{n}(\bar{\pi}_{i}-1/n)x_{i}(0)$, where $%
\bar{\pi}$ is the expected value of $\pi $ and $\bar{\pi}_{i}$ denotes its $%
i $th component. The greater is this gap, the greater is the potential for
misinformation in this society. Moreover, this formula also makes it clear
that $\bar{\pi}_{i}-1/n$ gives the \emph{excess influence} of agent $i$. Our
strategy will be to develop bounds on the spread of misinformation in the
society (as defined above) and on the excess influence of each agent for
general social networks and also provide exact results for some special
networks.

We provide three types of results. First, using tools from matrix
perturbation
theory,\footnote{%
In particular, we decompose the transition matrix of the Markov chain into a
doubly stochastic matrix, representing the underlying social network, and a
remainder matrix, representing a directed influence graph. Despite the term
\textquotedblleft perturbation,\textquotedblright\ this remainder matrix
need not be \textquotedblleft small\textquotedblright\ in any sense.} we
provide global and general upper bounds on the extent of misinformation as a
function of the properties of the underlying social network. In particular,
the bounds relate to the \emph{spectral gap} and the mixing properties of
the Markov chain induced by the social network. Recall that a Markov chain
is \textit{fast-mixing} if it converges rapidly to its stationary
distribution. It will do so when it has a large spectral gap, or loosely
speaking, when it is highly connected and possesses many potential paths of
communication between any pair of agents. Intuitively, societies represented
by fast-mixing Markov chains have more limited room for misinformation
because forceful agents themselves are influenced by the weighted opinion of
the rest of the society before they can spread their own (potentially
extreme) views. A corollary of these results is that for a special class of
societies, corresponding to \textquotedblleft expander
graphs\textquotedblright , misinformation disappears in large societies
provided that there is a finite number of forceful agents and no forceful
agent has global impact.\footnote{%
Expander graphs are graphs whose spectral gap remains bounded away
from zero as the number of nodes tends to infinity. Several networks
related to the Internet correspond to expander graphs; see, for
example, Mihail, Papadimitriou, and Saberi \cite{MPS}.} In contrast,
the extent of misinformation can be substantial in slow-mixing
Markov chains, also for an intuitive reason. Societies represented
by such Markov chains would have a high degree of partitioning
(multiple clusters with weak communication in between), so that
forceful agents receive their information from others who previously
were influenced by them, ensuring that their potentially extreme
opinions are never moderated.\footnote{%
This result is related to Golub and Jackson \cite{golub-two}, where they
relate learning to homophily properties of the social network.}

Our second set of results exploit the local structure of the
social network in the neighborhood of each forceful agent in order
to provide a tighter characterization of the extent of
misinformation and excess influence. Fast-mixing and spectral gap
properties are global (and refer to the properties of the overall
social network representing meeting and communication patterns
among all agents). As such, they may reflect properties of a
social network far from where the forceful agents are located. If
so, our first set of bounds will not be tight. To redress this
problem, we develop an alternative analysis using mean (first)
passage times of the Markov chain and show how it is not only the
global properties of the social network, but also the local social
context in which forceful agents are situated that matter. For
example, in a social network with a single dense cluster and
several non-clustered pockets, it matters greatly whether forceful
links are located inside the cluster or not. We illustrate this
result sharply by first focusing on graphs with \emph{forceful
essential edges}, that is, graphs representing societies in which
a single forceful link connects two otherwise disconnected
components. This, loosely speaking, represents a situation in
which a forceful agent, for example a media outlet or a political
party leader, obtains all of its (or his or her) information from
a small group of individuals and influences the rest of the
society. In this context, we establish the surprising result that
all members of the small group will have the same excess
influence, even though some of them may have much weaker links or
no links to the forceful agent. This result is an implication of
the society having a (single) forceful essential edge and reflects
the fact that the information among the small group of individuals
who are the source of information of the forceful agent aggregates
rapidly and thus it is\ the average of their beliefs that matter.
We then generalize these results and intuitions to more general
graphs using the notion of \emph{information bottlenecks.}

Our third set of results are more technical in nature, and provide new
conceptual tools and algorithms for characterizing the role of information
bottlenecks. In particular, we introduce the concept of \emph{relative cuts}
and present several new results related to relative cuts and how these
relate to mean first passage times. For our purposes, these new results are
useful because they enable us to quantify the extent of local clustering
around forceful agents. Using the notion of relative cuts, we develop new
algorithms based on graph clustering that enable us to provide improved
bounds on the extent of misinformation in beliefs as a function of
information bottlenecks in the social network.

Our paper is related to a large and growing learning literature. Much of
this literature focuses on various Bayesian models of observational or
communication-based learning; for example Bikchandani, Hirshleifer and Welch
\cite{BHW}, Banerjee \cite{abhijit}, Smith and Sorensen \cite{SS}, \cite%
{SSnew}, Banerjee and Fudenberg \cite{ban-fud}, Bala and Goyal \cite%
{bala-goyal}, \cite{bala-goyal2}, Gale and Kariv \cite{gale-kariv}, and
Celen and Kariv \cite{celen-kariv}, \cite{celen-kariv2}. These papers
develop models of social learning either using a Bayesian perspective or
exploiting some plausible rule-of-thumb behavior. Acemoglu, Dahleh, Lobel
and Ozdaglar \cite{ilan} provide an analysis of Bayesian learning over
general social networks. Our paper is most closely related to DeGroot \cite%
{degroot}, DeMarzo, Vayanos and Zwiebel \cite{demarzo} and Golub and Jackson
\cite{golub}, \cite{golub-two}, who also consider non-Bayesian learning over
a social network represented by a connected graph.\footnote{%
An important distinction is that in contrast to the \textquotedblleft
averaging\textquotedblright\ model used in these papers, we have a model of
pairwise interactions. We believe that this model has a more attractive
economic interpretation, since it does not have the feature that neighbors'
information will be averaged at each date (even though the same information
was exchanged the previous period). In contrast, in the pairwise meeting
model (without forceful agents), if a pair meets two periods in a row, in
the second meeting there is no information to exchange and no change in
beliefs takes place.} None of the papers mentioned above consider the issue
of the spread of misinformation (or the tension between aggregation of
information and spread of misinformation), though there are close parallels
between Golub and Jackson's and our characterizations of influence.\footnote{%
In particular, Golub and Jackson \cite{golub-two} characterize the
effects of homophily on learning and influence in two different
models of learning in terms of mixing properties and the spectral
gap of graphs. In one of their learning models, which builds on
DeGroot \cite{degroot}, DeMarzo, Vayanos and Zwiebel \cite{demarzo}
and Golub and Jackson \cite{golub}, homophily has negative effects
on learning (and speed of learning) for reasons related to our
finding that in slow-mixing graphs, misinformation can spread more.}
In addition to our focus, the methods of analysis here, which
develop bounds on the extent of misinformation and provide exact
characterization of excess influence in certain classes of social
networks, are entirely new in the literature and also rely on the
developments of new results in the analysis of Markov chains.

Our work is also related to other work in the economics of
communication, in particular, to cheap-talk models based on Crawford
and Sobel \cite{crawford-sobel} (see also Farrell and Gibbons
\cite{farrell} and Sobel \cite{sobel}), and some recent learning
papers incorporating cheap-talk games into a network structure (see
Ambrus and Takahashi \cite{ambrus}, Hagenbach and
Koessler \cite{hagenbach}, and Galeotti, Ghiglino and Squintani \cite%
{galeotti}).

In addition to the papers on learning mentioned above, our paper is related
to work on consensus, which is motivated by different problems, but
typically leads to a similar mathematical formulation (Tsitsiklis \cite%
{johnthes}, Tsitsiklis, Bertsekas and Athans \cite{distasyn}, Jadbabaie, Lin
and Morse \cite{ali}, Olfati-Saber and Murray \cite{reza}, Olshevsky and
Tsitsiklis \cite{alexlong}, Nedi\'c and Ozdaglar \cite{distpaper}). In
consensus problems, the focus is on whether the beliefs or the values held
by different units (which might correspond to individuals, sensors or
distributed processors) converge to a common value. Our analysis here does
not only focus on consensus, but also whether the consensus happens around
the true value of the underlying state. There are also no parallels in this
literature to our bounds on misinformation and characterization results.

The rest of this paper is organized as follows: In Section \ref{model}, we
introduce our model of interaction between the agents and describe the
resulting evolution of individual beliefs. We also state our assumptions on
connectivity and information exchange between the agents. Section \ref%
{convergence} presents our main convergence result on the evolution of agent
beliefs over time. In Section \ref{global-misinfo}, we provide bounds on the
extent of misinformation as a function of the global network parameters.
Section \ref{local-misinfo} focuses on the effects of location of forceful
links on the spread of misinformation and provides bounds as a function of
the local connectivity and location of forceful agents in the network.
Section \ref{conclusions} contains our concluding remarks.

\vskip 1pc

\noindent \textbf{Notation and Terminology:} A vector is viewed as a column
vector, unless clearly stated otherwise. We denote by $x_i$ or $[x]_i$ the $%
i^{th}$ component of a vector $x$. When $x_i\ge 0$ for all components $i$ of
a vector $x$, we write $x\ge 0$. For a matrix $A$, we write $A_{ij}$ or $%
[A]_{ij}$ to denote the matrix entry in the $i^{th}$ row and $j^{th}$
column. We write $x^{\prime }$ to denote the transpose of a vector $x$. The
scalar product of two vectors $x,y\in\mathbb{R}^m$ is denoted by $x^{\prime
}y$. We use $\|x\|_2$ to denote the standard Euclidean norm, $\|x\|_2=\sqrt{%
x^{\prime }x}$. We write $\|x\|_\infty$ to denote the max norm, $%
\|x\|_\infty = \max_{1\le i\le m}|x_i|$. We use $e_i$ to denote the vector
with $i^{th}$ entry equal to 1 and all other entries equal to 0. We denote
by $e$ the vector with all entries equal to 1.

A vector $a$ is said to be a \textit{stochastic vector} when $a_i\ge0$ for
all $i$ and $\sum_i a_i =1$. A square matrix $A$ is said to be a \textit{%
(row) stochastic matrix} when each row of $A$ is a stochastic vector. The
transpose of a matrix $A$ is denoted by $A^{\prime }$. A square matrix $A$
is said to be a \textit{doubly stochastic matrix} when both $A$ and $%
A^{\prime }$ are stochastic matrices.

\section{Belief Evolution}

\label{model}

\subsection{Description of the Environment}

We consider a set $\mathcal{N}=\{1,\ldots,n\}$ of agents interacting over a
social network. Each agent $i$ starts with an initial belief about an
underlying state, which we denote by $x_i(0)\in \mathbb{R}$. Agents exchange
information with their neighbors and update their beliefs. We assume that
there are two types of agents; \textit{regular and forceful}. Regular agents
exchange information with their neighbors (when they meet). In contrast,
forceful agents influence others disproportionately.

We use an asynchronous continuous-time model to represent meetings between
agents (also studied in Boyd \textit{et al.} \cite{boyd} in the context of
communication networks). In particular, we assume that each agent meets
(communicates with) other agents at instances defined by a rate one Poisson
process independent of other agents. This implies that the meeting instances
(over all agents) occur according to a rate $n$ Poisson process at times $%
t_{k}$, $k\geq 1.$ Note that in this model, by convention, at most
one node is active (i.e., is meeting another) at a given time. We
discretize time according to meeting instances (since these are the
relevant instances at which the beliefs change), and refer to the
interval $[t_{k},t_{k+1})$ as the \textit{$k^{th}$ time slot}. On
average, there are $n$ meeting instances per unit of absolute time
(see Boyd \textit{et al.} \cite{boyd} for a precise relation between
these instances and absolute time). Suppose that at time (slot) $k$,
agent $i$ is chosen to meet another agent (probability $1/n$). In
this case, agent $i$ will meet agent $j\in \mathcal{N}$ with probability $%
p_{ij}$. 
Following a meeting between $i$ and $j$, there is a potential exchange of
information. Throughout, we assume that all events that happen in a meeting
are \textit{independent of any other event that happened in the past}. Let $%
x_{i}(k)$ denote the belief of agent $i$ about the underlying state
at time $k$. The agents update their beliefs according to one of the
following three possibilities.

\begin{itemize}
\item[(i)] Agents $i$ and $j$ reach pairwise consensus and the beliefs are
updated according to
\begin{equation*}
x_{i}(k+1)=x_{j}(k+1)=\frac{x_{i}(k)+x_{j}(k)}{2}.
\end{equation*}%
We denote the conditional probability of this event (conditional on $i$
meeting $j$) as $\beta _{ij}$.%

\item[(ii)] Agent $j$ influences agent $i$, in which case for some $\epsilon
\in (0,1/2]$, beliefs change according to
\begin{equation}
x_{i}(k+1)=\epsilon x_{i}(k)+(1-\epsilon )x_{j}(k),\quad \hbox{and}\quad
x_{j}(k+1)=x_{j}(k).  \label{selfweight}
\end{equation}%
In this case beliefs of agent $j$ do not change.\footnote{%
We could allow the self belief weight $\epsilon $ to be different for each
agent $i$. This generality does not change the results or the economic
intuitions, so for notational convenience, we assume this weight to be the
same across all agents.} We denote the conditional probability of this event
as $\alpha _{ij}$, and refer to it as the \textit{influence probability}.
Note that we allow $\epsilon =1/2$, so that agent $i$ may be treating agent $%
j$ just as a regular agent, except that agent $j$ himself does not change
his beliefs.

\item[(iii)] Agents $i$ and $j$ do not agree and stick to their beliefs,
i.e.,
\begin{equation*}
x_{i}(k+1)=x_{i}(k),\quad \hbox{and}\quad x_{j}(k+1)=x_{j}(k).
\end{equation*}%
This event has probability $\gamma _{ij}=1-\beta _{ij}-\alpha _{ij}$.
\end{itemize}

Any agent $j$ for whom the influence probability $\alpha _{ij}>0$ for some $%
i\in \mathcal{N}$ is referred to as a \textit{forceful agent}.
Moreover, the directed link $(j,i)$ is referred to as a
\textit{forceful link}.\footnote{We refer to directed links/edges as
links and undirected ones as edges.}

As discussed in the introduction, we can interpret forceful agents in
multiple different ways. First, forceful agents may correspond to community
leaders or news media, will have a disproportionate effect on the beliefs of
their followers. In such cases, it is natural to consider $\epsilon$ small
and the leaders or media not updating their own beliefs as a result of
others listening to their opinion. Second, forceful agents may be
indistinguishable from regular agents, and thus regular agents engage in
what they think is information exchange, but forceful agents, because of
stubbornness or some other motive, do not incorporate the information of
these agents in their own beliefs. In this case, it may be natural to think
of $\epsilon $ as equal to $1/2$. The results that follow remain valid with
either interpretation.

The influence structure described above will determine the evolution of
beliefs in the society. Below, we will give a more precise separation of
this evolution into two components, one related to the underlying social
network (communication and meeting structure), and the other to influence
patterns.

\subsection{Assumptions}\label{assumptions}

We next state our assumptions on the belief evolution model among
the agents. We have the following assumption on the agent meeting
probabilities $p_{ij}$.

\begin{assumption}
(Meeting Probabilities) \emph{
\begin{itemize}
\item[(a)] For all $i$, the probabilities $p_{ii}$ are equal to 0.
\item[(b)] For all $i$, the probabilities $p_{ij}$ are nonnegative
for all $j$ and they sum to 1 over $j$, i.e.,
\[p_{ij}\ge 0\quad \hbox{for all } i,j,\qquad \sum_{j=1}^n p_{ij} = 1\quad \hbox{for all }i.\]
\end{itemize}}\label{comprob}
\end{assumption}

Assumption \ref{comprob}(a) imposes that \textquotedblleft
self-communication\textquotedblright\ is not a possibility, though this is
just a convention, since, as stated above, we allow disagreement among
agents, i.e., $\gamma _{ij}$ can be positive. We let $P$ denote the matrix
with entries $p_{ij}$. Under Assumption \ref{comprob}(b), the matrix $P$ is
a \emph{stochastic matrix.}\footnote{%
That is, its row sums are equal to 1.}

We next impose a connectivity assumption on the social network. This
assumption is stated in terms of the directed graph $(\mathcal{N},\mathcal{E}%
)$, where $\mathcal{E}$ is the set of directed links induced by the
positive meeting probabilities $p_{ij}$, i.e.,
\begin{equation}
\mathcal{E}=\{(i,j)\ |\ p_{ij}>0\}.  \label{edges}
\end{equation}

\begin{assumption}
(Connectivity) \emph{The graph $(\mathcal{N},\mathcal{E})$ is
strongly connected, i.e., for all $i,j \in \mathcal{N}$, there
exists a directed path connecting $i$ to $j$ with links in the set
$\mathcal{E}$.}\label{connect}
\end{assumption}

Assumption \ref{connect} ensures that every agent \textquotedblleft
communicates\textquotedblright\ with every other agent (possibly through
multiple links). This is not an innocuous assumption, since otherwise the
graph $(\mathcal{N},\mathcal{E})$ (and the society that it represents) would
segment into multiple non-communicating parts. Though not innocuous, this
assumption is also natural for several reasons. First, the evidence suggests
that most subsets of the society are not only connected, but are connected
by means of several links (e.g., Watts \cite{watts} and Jackson \cite%
{Jacksonbook}), and the same seems to be true for indirect linkages via the
Internet. Second, if the society is segmented into multiple
non-communication parts, the insights here would apply, with some
modifications, to each of these parts.

Let us also use $d_{ij}$ to denote the length of the shortest path from $i$
to $j$ and $d$ to denote the \textit{maximum shortest path length} between any $%
i,j\in \mathcal{N}$, i.e.,
\begin{equation}
d=\max_{i,j\in \mathcal{N}}d_{ij}.  \label{maxsp}
\end{equation}
In view of Assumption \ref{connect}, these are all well-defined objects.

Finally, we introduce the following assumption which ensures that there is
positive probability that every agent (even if he is forceful) receives some
information from an agent in his neighborhood.

\begin{assumption}
(Interaction Probabilities) \emph{For all $(i,j)\in \mathcal{E}$, the sum of
the \textit{averaging probability} $\beta_{ij}$ and the \textit{influence
probability} $\alpha_{ij}$ is positive, i.e.,
\begin{equation*}
\beta_{ij}+\alpha_{ij}>0\qquad \hbox{for all } (i,j)\in \mathcal{E}.
\end{equation*}%
} \label{intprob}
\end{assumption}

The connectivity assumption (Assumption \ref{connect}) ensures that
there is a path from any forceful agent to other agents in the
network, implying that for any forceful agent $i$, there is a link
$(i,j)\in \mathcal{E}$ for some $j\in \mathcal{N}$. Then the main
role of Assumption \ref{intprob} is to guarantee that even the
forceful agents at some point get information from
the other agents in the network.\footnote{%
This assumption is stated for all $(i,j)\in \mathcal{E}$, thus a
forceful agent $i$ receives some information from any $j$ in his
\textquotedblleft neighborhood". This is without any loss of
generality, since we can always
set $p_{ij}=0$ for those $j$'s that are in $i$'s neighborhood but from whom $%
i$ never obtains information.} This assumption captures the idea
that \textquotedblleft no man is an island,\textquotedblright\ i.e.,
even the beliefs of forceful agents are affected by the beliefs of
the society. In the absence of this assumption, any society
consisting of several forceful agents may never settle into a
stationary distribution of beliefs. While this is an interesting
situation to investigate, it requires a very different approach.
Since we view the \textquotedblleft no man is an
island\textquotedblright\ feature plausible, we find Assumption
\ref{intprob} a useful starting point.

Throughout the rest of the paper, we assume that Assumptions \ref{comprob}, %
\ref{connect}, and \ref{intprob} hold.

\subsection{Evolution of Beliefs: Social Network and Influence Matrices}

We can express the preceding belief update model compactly as follows. Let $%
x(k)=(x_1(k),\ldots,x_n(k))$ denote the vector of agent beliefs at time $k$.
The agent beliefs are updated according to the relation
\begin{equation}
x(k+1) = W(k)x(k),  \label{beliefupdate}
\end{equation}
where $W(k)$ is a random matrix given by
\begin{equation}
W(k) = \left\{
\begin{array}{cc}
A_{ij} \equiv I-\frac{(e_i-e_j)(e_i-e_j)^{\prime }}{2} & \text{with
probability } p_{ij}\beta_{ij}/n, \\
J_{ij} \equiv I- (1-\epsilon)e_i (e_i-e_j)^{\prime } & \text{with
probability } p_{ij}\alpha_{ij}/n, \\
I & \text{with probability } p_{ij}\gamma_{ij}/n,%
\end{array}%
\right.  \label{updatematrix}
\end{equation}
for all $i,j\in \mathcal{N}$. The preceding belief update model implies that
the matrix $W(k)$ is a stochastic matrix for all $k$, and is independent and
identically distributed over all $k$.

Let us introduce the matrices
\begin{equation}
\Phi (k,s)=W(k)W(k-1)\cdots W(s+1)W(s)\qquad \hbox{for all $k$ and $s$ with }%
\ k\geq s,  \label{transition-mat}
\end{equation}%
with $\Phi (k,k)=W(k)$ for all $k$. We will refer to the matrices $\Phi
(k,s) $ as the \textit{transition matrices}. We can now write the belief
update rule (\ref{beliefupdate}) as follows: for all $s$ and $k$ with $k\geq
s\geq 0 $ and all agents $i\in \{1,\ldots ,n\}$,
\begin{equation}
x_{i}(k+1)=\sum_{j=1}^{n}[\Phi (k,s)]_{ij}\,x_{j}(s).  \label{beliefinterval}
\end{equation}

Given our assumptions, the random matrix $W(k)$ is identically
distributed over all $k$, and thus we have for some nonnegative
matrix $\tilde{W}$,
\begin{equation}
E[W(k)]=\tilde{W}\qquad \hbox{for all }k\geq 0.  \label{meanint}
\end{equation}%
The matrix, $\tilde{W}$, which we refer to as the \textit{mean interaction
matrix}, represents the evolution of beliefs in the society. It incorporates
elements from both the underlying social network (which determines the
meeting patterns) and the influence structure. In what follows, it will be
useful to separate these into two components, both for our mathematical
analysis and to clarify the intuitions. For this purpose, let us use the
belief update model (\ref{beliefupdate})-(\ref{updatematrix}) and write the
mean interaction matrix $\tilde{W}$ as follows:\footnote{%
In the sequel, the notation $\sum_{i,j}$ will be used to denote the double
sum $\sum_{i=1}^{n}\sum_{j=1}^{n}$.}
\begin{eqnarray*}
\tilde{W} &=&{\frac{1}{n}}\sum_{i,j}p_{ij}\Big[\beta _{ij}A_{ij}+\alpha
_{ij}J_{ij}+\gamma _{ij}I\Big] \\
&=&{\frac{1}{n}}\sum_{i,j}p_{ij}\Big[(1-\gamma _{ij})A_{ij}+\gamma _{ij}I%
\Big]+{\frac{1}{n}}\sum_{i,j}p_{ij}\alpha _{ij}\big[J_{ij}-A_{ij}\big],
\end{eqnarray*}%
where $A_{ij}$ and $J_{ij}$ are matrices defined in Eq.\ (\ref{updatematrix}%
), and the second inequality follows from the fact that $\beta
_{ij}=1-\alpha _{ij}-\gamma _{ij}$ for all $i,j\in \mathcal{N}$. We use the
notation
\begin{equation}
T={\frac{1}{n}}\sum_{i,j}p_{ij}\Big[(1-\gamma _{ij})A_{ij}+\gamma _{ij}I\Big]%
,\quad D={\frac{1}{n}}\sum_{i,j}p_{ij}\alpha _{ij}\big[J_{ij}-A_{ij}\big],
\label{decomposition}
\end{equation}%
to write the mean interaction matrix, $\tilde{W}$, as
\begin{equation}
\tilde{W}=T+D.  \label{intdecomp}
\end{equation}

Here, the matrix $T$ only depends on meeting probabilities (matrix $P$)
except that it also incorporates $\gamma _{ij}$ (probability that following
a meeting no exchange takes place). We can therefore think of the matrix $T$
as representing the underlying \emph{social network} (friendships,
communication among coworkers, decisions about which news outlets to watch,
etc.), and refer to it as the \emph{social network matrix}. It will be
useful below to represent the social interactions using an undirected (and
weighted) graph induced by the social network matrix $T$. This graph is
given by $(\mathcal{N,A})$, where $\mathcal{A}$ is the set of undirected
edges given by
\begin{equation}
\mathcal{A}=\Big\{\{i,j\}\ |\ T_{ij}>0\Big \},  \label{undirgraph}
\end{equation}%
and the weight $w_{e}$ of edge $e=\{i,j\}$ is given by the entry $%
T_{ij}=T_{ji}$ of the matrix $T$. We refer to this graph as the \textit{%
social network graph}.

The matrix $D$, on the other hand, can be thought of as representing
the \emph{influence structure} in the society. It incorporates
information about which individuals and links are forceful (i.e.,
which types of interactions will lead to one individual influencing
the other without updating his own beliefs). We refer to matrix $D$
as the \emph{influence matrix}. It is also useful to note for
interpreting the mathematical results below that $T$ is a doubly
stochastic matrix, while $D$ is not. Therefore, Eq.
(\ref{intdecomp}) gives a decomposition of the mean connectivity
matrix $\tilde{W}$ into a doubly stochastic and a remainder
component, and enables us to use tools from matrix perturbation
theory (see Section \ref{global-misinfo}).

\section{Convergence\label{convergence}}

In this section, we provide our main convergence result. In particular, we
show that despite the presence of forceful agents, with potentially very
different opinions at the beginning, the society will ultimately converge to
a consensus, in which all individuals share the same belief. This consensus
value of beliefs itself is a random variable. We also provide a first
characterization of the expected value of this consensus belief in terms of
the mean interaction matrix (and thus social network and influence
matrices). Our analysis essentially relies on showing that iterates of Eq.\ (%
\ref{beliefupdate}), $x(k)$, converge to a \textit{consensus} with
probability one, i.e., $x(k)\rightarrow \bar{x}e$, where $\bar{x}$ is a
scalar random variable that depends on the initial beliefs and the random
sequence of matrices $\{W(k)\}$, and $e$ is the vector of all one's. The
proof uses two lemmas which are presented in Appendix B.

\begin{theorem}
The sequences $\{x_i(k)\}$, $i\in \mathcal{N}$, generated by Eq.\ (\ref%
{beliefupdate}) converge to a consensus belief, i.e., there exists a scalar
random variable $\bar{x}$ such that
\begin{equation*}
\lim_{k\to \infty} x_i(k)=\bar x\qquad \hbox{for all }i
\hbox{ with
probability one}.
\end{equation*}
Moreover, the random variable $\bar x$ is a convex combination of initial
agent beliefs, i.e.,
\begin{equation*}
\bar x = \sum_{j=1}^n \pi_j x_j(0),
\end{equation*}
where $\pi=[\pi_1,\ldots,\pi_n]$ is a random vector that satisfies $\pi_j\ge
0$ for all $j$, and $\sum_{j=1}^n \pi_j = 1$. \label{convconsensus}
\end{theorem}

\begin{proof}
By Lemma \ref{posprob} from Appendix B, we have
\[P\left\{ [\Phi(s+ n^2d -1,s)]_{ij}\ge {\eta^d\over 2} \epsilon^{n^2-1},\
\hbox{for all }i,j\right\} \ge
\left(\frac{\eta^d}{2}\right)^{n^2}\qquad \hbox{for all }s\ge 0,
\]
where $\Phi(s+ n^2d -1,s)$ is a transition matrix [cf.\ Eq.\
(\ref{transition-mat})], $d$ is the maximum shortest path length in
graph $({\cal N,E})$ [cf.\ Eq.\ (\ref{maxsp})], $\epsilon$ is the
self belief weight against a forceful agent [cf.\ Eq.\
(\ref{selfweight})], and $\eta$ is a positive scalar defined in Eq.\
(\ref{scalareta}). This relation implies that over a window of
length $n^2d$, all entries of the transition matrix $\Phi(s+ n^2d
-1,s)$ are strictly positive with positive probability, which is
uniformly bounded away from 0. Thus, we can use Lemma \ref{lyapdec}
(from Appendix A) with the identifications
\[H(k)=W(k), \qquad B=n^2d,\qquad \theta={\eta^d\over 2} \epsilon^{n^2-1}.\]
Letting
\[M(k)=\max_{i\in {\cal N}} x_i(k),\qquad m(k)=\min_{i\in {\cal N}} x_i(k),\]
this implies that $n {\eta^d\over 2} \epsilon^{n^2-1}\le 1$ and for
all $s\ge 0$,
\[P\left\{ M(s+n^2d)-m(s+n^2d)\le (1-n\eta^d/2\, \epsilon^{n^2-1})
(M(s)-m(s))\right\} \ge \left(\frac{\eta^d}{2}\right)^{n^2}.
\]
Moreover, by the stochasticity of the matrix $W(k)$, it follows that
the sequence $\{M(k)-m(k)\}$ is nonincreasing with probability one.
Hence, we have for all $s\ge 0$
\[E\Big[M(s+n^2d)-m(s+n^2d) \Big] \le \left[1-\left(\frac{\eta^d}{2}\right)^{n^2} +
\left(\frac{\eta^d}{2}\right)^{n^2} (1-n\eta^d/2\,
\epsilon^{n^2-1})\right] (M(s)-m(s)),\] from which, for any $k\ge
0$, we obtain
\[E\Big[M(k)-m(k) \Big] \le \left[1-\left(\frac{\eta^d}{2}\right)^{n^2} +
\left(\frac{\eta^d}{2}\right)^{n^2} (1-n\eta^d/2\,
\epsilon^{n^2-1})\right]^{\lfloor{k\over n^2d }\rfloor}
(M(0)-m(0)).\] This implies that
\[\lim_{k\to \infty} M(k)-m(k)=0\qquad \hbox{with probability one}.\]
The stochasticity of the matrix $W(k)$ further implies that the
sequences $\{M(k)\}$ and $\{m(k)\}$ are bounded and monotone and
therefore converges to the same limit, which we denote by $\bar x$.
Since we have
\[m(k)\le x_i(k)\le M(k)\qquad \hbox{for all }i \hbox{ and }k\ge 0,\]
it follows that
\[\lim_{k\to \infty} x_i(k)=\bar x\qquad \hbox{for all }i \hbox{ with probability one},\]
establishing the first result.

Letting $s=0$ in Eq.\ (\ref{beliefinterval}), we have for all $i$
\begin{equation}x_i(k) = \sum_{j=1}^n [\Phi(k-1,0)]_{ij}\,
x_j(0)\qquad \hbox{for all } k\ge 0.\label{concomb}\end{equation}
From the previous part, for any initial belief vector $x(0)$, the
limit
\[\lim_{k\to \infty} x_i(k) = \sum_{j=1}^n \lim_{k\to \infty}[\Phi(k-1,0)]_{ij}\, x_j(0)\]
exists and is independent of $i$. Hence, for any $h$, we can choose
$x(0)=e_h$, i.e., $x_h(0)=1$ and $x_j(0)=0$ for all $j\ne h$,
implying that the limit
\[\lim_{k\to \infty}[\Phi(k-1,0)]_{ih}\]
exists and is independent of $i$. Denoting this limit by $\pi_h$ and
using Eq.\ (\ref{concomb}), we obtain the desired result, where the
properties of the vector $\pi=[\pi_1,\ldots,\pi_n]$ follows from the
stochasticity of matrix $\Phi(k,0)$ for all $k$ (implying the
stochasticity of its limit as $k\to \infty$).
\end{proof}

The key implication of this result is that, despite the presence of forceful
agents, the society will ultimately\ reach a consensus. Though surprising at
first, this result is intuitive in light of our \textquotedblleft no man is
an island\textquotedblright\ assumption (Assumption \ref{intprob}). However,
in contrast to \textquotedblleft averaging models\textquotedblright\ used
both in the engineering literature and recently in the learning literature,
the consensus value here is a random variable and will depend on the order
in which meetings have taken place. The main role of this result for us is
that we can now conduct our analysis on quantifying the extent of the spread
of misinformation by looking at this consensus value of beliefs.

The next theorem characterizes $E[\bar{x}]$ in terms of the limiting
behavior of the matrices $\tilde{W}^{k}$ as $k$ goes to infinity.

\begin{theorem}
Let $\bar{x}$ be the limiting random variable of the sequences $\{x_i(k)\}$,
$i\in \mathcal{N}$ generated by Eq.\ (\ref{beliefupdate}) (cf.\ Theorem \ref%
{convconsensus}). Then we have:
\begin{itemize}
\item[(a)] The matrix $\tilde{W}^{k}$ converges to a stochastic matrix with
identical rows $\bar{\pi}$ as $k$ goes to infinity, i.e.,
\begin{equation*}
\lim_{k\rightarrow \infty }\tilde{W}^{k}=e\bar{\pi}^{\prime }.
\end{equation*}
\item[(b)] The expected value of $\bar{x}$ is given by a convex combination
of the initial agent values $x_i(0)$, where the weights are given by the
components of the probability vector $\bar{\pi}$, i.e.,
\begin{equation*}
E[\bar{x}]= \sum_{i=1}^n \bar{\pi}_i x_i(0)= \bar{\pi}^{\prime }x(0).
\end{equation*}
\end{itemize}\label{tildewlim}
\end{theorem}

\begin{proof} (a) \ This part relies on the properties of the mean
interaction matrix established in Appendix B. In particular, by
Lemma \ref{tranprimitive}(a), the mean interaction matrix $\tilde W$
is a primitive matrix. Therefore, the Markov Chain with transition
probability matrix $\tilde W$ is regular (see Section
\ref{prelim-results} for a definition). The result follows
immediately from Theorem \ref{MC}(a).

\noindent (b) \ From Eq.\ (\ref{beliefinterval}), we have for all
$k\ge 0$
\[x(k)= \Phi(k-1,0)x(0).\]
Moreover, since $x(k)\to \bar x e$ as $k\to \infty$, we have
\[E[\bar{x} e]= E[\lim_{k\to \infty} x(k)] = \lim_{k\to \infty} E[x(k)],\]
where the second equality follows from the Lebesgue's Dominated
Convergence Theorem (see \cite{rudin}). Combining the preceding two
relations and using the assumption that the matrices $W(k)$ are
independent and identically distributed over all $k\ge 0$, we obtain
\[E[\bar{x} e] = \lim_{k\to \infty}  E[\Phi(k-1,0)x(0)] =\lim_{k\to \infty} \tilde W^k x(0),\]
which in view of part (a) implies
\[E[\bar{x}]=  \bar{\pi}'x(0).\]
\end{proof}

Combining Theorem \ref{convconsensus} and Theorem \ref{tildewlim}(a)
(and using the fact that the results hold for any $x(0)$), we have
$\bar \pi= E[\pi]$. The stationary distribution $\bar{\pi}$ is
crucial in understanding the formation of opinions since it
encapsulates the weight given to each agent (forceful or regular) in
the (limiting) mean consensus
value of the society. We refer to the vector $\bar{\pi}$ as the \textit{%
consensus distribution} corresponding to the mean interaction matrix $\tilde{%
W}$ and its component $\bar{\pi}_{i}$ as the \textit{weight} of agent $i$.

It is also useful at this point to highlight how consensus will form around
the correct value in the absence of forceful agents. Let $\{x(k)\}$ be the
belief sequence generated by the belief update rule of Eq.\ (\ref%
{beliefupdate}). When there are no forceful agents, i.e. $\alpha
_{ij}=0$ for all $i,j$, then the interaction matrix $W(k)$ for all
$k$ is either equal to an averaging matrix $A_{ij}$ for some $i, j$
or equal to the identity matrix $I$; hence, $W(k)$ is a doubly
stochastic matrix. This implies that the average value of $x(k)$
remains constant at each iteration, i.e.,
\begin{equation*}
{\frac{1}{n}}\sum_{i=1}^{n}x_{i}(k)={\frac{1}{n}}\sum_{i=1}^{n}x_{i}(0)%
\qquad \hbox{for all }k\geq 0.
\end{equation*}

Theorem \ref{convconsensus} therefore shows that when there are no forceful
agents, the sequences $x_{i}(k)$ for all $i$, converge to the average of the
initial beliefs with probability one, aggregating information. We state this
result as a simple corollary.

\begin{corollary}
\label{corollary} \emph{Assume that there are no forceful agents, i.e., $%
\alpha_{ij}=0$ for all $i,j \in \mathcal{N}$. We have
\begin{equation*}
\lim_{k\rightarrow \infty }x_{i}(k)={\frac{1}{n}}\sum_{i=1}^{n}x_{i}(0)=%
\theta \qquad \hbox{with probability one}.
\end{equation*}}
\end{corollary}

Therefore, in the absence of forceful agents, the society is able to
aggregate information effectively. Theorem \ref{tildewlim} then also
implies that in this case $\pi=\bar{\pi}_{i}=1/n$ for all $i$ (i.e.,
beliefs converge to a deterministic value), so that no individual
has excess influence. These results no longer hold when there are
forceful agents. In the next section, we investigate the effect of
the forceful agents and the structure of the social network on the
extent of misinformation and excess influence of individuals.

\section{Global Limits on Misinformation\label{global-misinfo}}

In this section, we are interested in providing an upper bound on
the expected value of the difference between the consensus belief
$\bar{x}$ (cf.\ Theorem \ref{convconsensus}) and the true underlying
state, $\theta $ (or equivalently the average of the initial
beliefs), i.e.,
\begin{equation}
E\left[ \bar{x}-\theta\right]
=E[\bar{x%
}]-\theta =\sum_{i\in \mathcal{N}}\Big(\bar{\pi}_{i}-{\frac{1}{n}}\Big)x_{i}(0),  \label{equation to be bounded}
\end{equation}%
(cf.\ Theorem \ref{tildewlim}). Our bound relies on a fundamental theorem
from the perturbation theory of finite Markov Chains. Before presenting the
theorem, we first introduce some terminology and basic results related to
Markov Chains.

\subsection{Preliminary Results}

\label{prelim-results}

Consider a finite Markov Chain with $n$ states and transition probability
matrix $T$.\footnote{%
We use the same notation as in (\ref{intdecomp}) here, given the close
connection between the matrices introduced in the next two theorems and the
ones in (\ref{intdecomp}).} We say that a finite Markov chain is \textit{%
regular} if its transition probability matrix is a primitive matrix, i.e.,
there exists some integer $k>0$ such that all entries of the power matrix $%
T^{k}$ are positive. The following theorem states basic results on the
limiting behavior of products of transition matrices of Markov Chains (see
Theorems 4.1.4, 4.1.6, and 4.3.1 in Kemeny and Snell \cite{kemeny}).

\begin{theorem}
Consider a regular Markov Chain with $n$ states and transition probability
matrix $T$.
\begin{itemize}
\item[(a)] The $k^{th}$ power of the transition matrix $T$, $T^k$, converges
to a stochastic matrix $T^{\infty}$ with all rows equal to the probability
vector $\pi$, i.e.,
\begin{equation*}
\lim_{k\to \infty} T^k = T^{\infty} = e \pi^{\prime },
\end{equation*}
where $e$ is the $n$-dimensional vector of all ones.
\item[(b)] The probability vector $\pi$ is a left eigenvector of the matrix $%
T$, i.e.,
\begin{equation*}
\pi^{\prime }T=\pi^{\prime }\quad \hbox{and}\quad \pi^{\prime }e = 1.
\end{equation*}
The vector $\pi$ is referred to as the \textit{stationary distribution} of
the Markov Chain.
\item[(c)] The matrix $Y=(I-T+T^{\infty})^{-1}-T^{\infty}$ is well-defined
and is given by
\begin{equation*}
Y=\sum_{k=0}^\infty (T^k-T^{\infty}).
\end{equation*}
The matrix $Y$ is referred to as the \textit{fundamental matrix} of the
Markov Chain.
\end{itemize}\label{MC}
\end{theorem}

The following theorem provides an exact perturbation result for the
stationary distribution of a regular Markov Chain in terms of its
fundamental matrix. The theorem is based on a result due to Schweitzer \cite%
{schweitzer} (see also Haviv and Van Der Heyden \cite{haviv}).

\begin{theorem}
Consider a regular Markov Chain with $n$ states and transition probability
matrix $T$. Let $\pi$ denote its unique stationary distribution and $Y$
denote its fundamental matrix. Let $D$ be an $n\times n$ perturbation matrix
such that the sum of the entries in each row is equal to 0, i.e.,
\begin{equation*}
\sum_{j=1}^n [D]_{ij} = 0\quad \hbox {for all }i.
\end{equation*}
Assume that the perturbed Markov chain with transition matrix $\hat T=T+D$
is regular. Then, the perturbed Markov chain has a unique stationary
distribution $\hat \pi$, and the matrix $I-DY$ is nonsingular. Moreover, the
change in the stationary distributions, $\rho=\hat \pi- \pi$, is given by
\begin{equation*}
\rho'=\pi' DY(I-DY)^{-1}.
\end{equation*}%
\label{pertbd}
\end{theorem}

\subsection{Main Results}

This subsection provides bounds on the difference between the consensus
distribution and the uniform distribution using the global properties of the
underlying social network. Our method of analysis will rely on the
decomposition of the mean interaction matrix $\tilde{W}$ given in (\ref%
{intdecomp}) into the social network matrix $T$ and the influence matrix $D$%
. Recall that $T$ is doubly stochastic.

The next theorem provides our first result on characterizing the
extent of misinformation and establishes an upper bound on the
$l_{\infty }$-norm of the difference between the stationary
distribution $\bar{\pi}$ and the uniform distribution
${\frac{1}{n}}e$, which, from Eq.\ (\ref{equation to be bounded}),
also provides a bound on the deviation between expected beliefs and
the true underlying state, $\theta $.

\begin{theorem}

\begin{itemize}
\item[(a)] Let $\bar{\pi}$ denote the consensus distribution. The $l_\infty$%
-norm of the difference between $\bar{\pi}$ and ${\frac{1}{n}}e$ is given by
\begin{equation*}
\Big\|\bar{\pi}-{\frac{1}{n}} e\Big\|_{\infty} \le {\frac{1}{1-\delta}} \, {%
\frac{\sum_{i,j} p_{ij} \alpha_{ij}}{2n}},
\end{equation*}
where $\delta$ is a constant defined by
\begin{equation*}
\delta = (1-n\chi^d)^{\frac{1}{d}},
\end{equation*}
\begin{equation*}
\chi=\min_{(i,j)\in \mathcal{E}} \left\{{\frac{1}{n}} \Big[p_{ij} \, {\frac{%
1-\gamma_{ij}}{2}} + p_{ji} \, {\frac{1-\gamma_{ji}}{2}} \Big]\right\},
\end{equation*}
and $d$ is the maximum shortest path length in the graph $(\mathcal{N},
\mathcal{E})$ [cf.\ Eq.\ (\ref{maxsp})].
\item[(b)] Let $\bar{x}$ be the limiting random variable of the sequences $%
\{x_i(k)\}$, $i\in \mathcal{N}$ generated by Eq.\ (\ref{beliefupdate}) (cf.\
Theorem \ref{convconsensus}). We have
\begin{equation*}
\Big|E[\bar{x}] - {\frac{1}{n}} \sum_{i=1}^n x_i(0)\Big| \le {\frac{1}{%
1-\delta}} \, {\frac{\sum_{i,j} p_{ij}\alpha_{ij}}{2n}}\, \|x(0)\|_{\infty}.
\end{equation*}
\end{itemize}

\label{delta_bd_thm}
\end{theorem}

\begin{proof} (a) \ Recall that the mean interaction matrix can be
represented as
\[\tilde W = T+D,\]
[cf.\ Eq.\ (\ref{intdecomp})], i.e., $\tilde W$ can be viewed as a
perturbation of the social network matrix $T$ by influence matrix
$D$. By Lemma \ref{dslim}(a), the stationary distribution of the
Markov chain with transition probability matrix $T$ is given by
the uniform distribution ${1\over n} e$. By the definition of the
matrix $D$ [cf.\ Eq.\ (\ref{decomposition})] and the fact that the
matrices $A_{ij}$ and $J_{ij}$ are stochastic matrices with all
row sums equal to one [cf.\ Eq.\ (\ref{updatematrix})], it follows
that the sum of entries of each row of $D$ is equal to 0.
Moreover, by Theorem \ref{tildewlim}(a), the Markov Chain with
transition probability matrix $\tilde W$ is regular and has  a
stationary distribution $\bar{\pi}$. Therefore, we can use the
exact perturbation result given in Theorem \ref{pertbd} to write
the change in the stationary distributions ${1\over n}e$ and
$\bar{\pi}$ as
\begin{equation}
 \Big(\bar{\pi} - {1\over n}e\Big)' = {1\over n}e' DY (I-DY)^{-1},\label{exactdif}
\end{equation}
where $Y$ is the fundamental matrix of the Markov Chain with
transition probability matrix $T$, i.e.,
\begin{equation}\label{fundmat_def}
Y=\sum_{k=0}^\infty (T^k-T^{\infty}), \nonumber
\end{equation}
with $T^\infty = {1\over n} e e'$ [cf.\ Theorem \ref{MC}(c)].
Algebraic manipulation of Eq.\ (\ref{exactdif}) yields
\[ \Big(\bar{\pi} - {1\over n}e\Big)' = \bar \pi' DY,\]
implying that
\begin{equation}
\Big\|\bar{\pi}- {1\over n}e\Big\|_{\infty} \le
\|DY\|_{\infty},\label{normdif}
\end{equation}
where $\|DY\|_{\infty}$ denotes the matrix norm induced by the
$l_{\infty}$ vector norm.

We next obtain an upper bound on the matrix norm $\|DY\|_{\infty}$.
By the definition of the fundamental matrix $Y$, we have
\begin{equation}DY = \sum_{k=0}^\infty D(T^k-T^{\infty}) =
\sum_{k=0}^\infty D T^k,\label{fundmat}\end{equation} where the
second equality follows from the fact that the row sums of matrix
$D$ is equal to 0 and the matrix $T^\infty$ is given by $T^\infty =
{1\over n} e e'$.

Given any $z(0)\in \mathbb{R}^n$ with $\|z(0)\|_{\infty}=1$, let
$\{z(k)\}$ denote the sequence generated by the linear update rule
\[z(k) = T^kz(0)\qquad \hbox{for all }k\ge 0.\]
Then, for all $k\ge 0$, we have
\[DT^k z(0) = Dz(k),\]
which by the definition of the matrix $D$ [cf.\ Eq.\
(\ref{decomposition})] implies \begin{equation}DT^k z(0) = {1\over
n} \sum_{i,j} p_{ij} \alpha_{ij}
z^{ij}(k),\label{infbd}\end{equation} where the vector $z^{ij}(k)\in
\mathbb{R}^n$ is defined as
\[z^{ij}(k) = [J_{ij}-A_{ij}] z(k)\qquad \hbox{for all }i,j, \hbox{ and }k\ge0.\]
By the definition of the matrices $J_{ij}$ and $A_{ij}$ [cf.\ Eq.\
(\ref{updatematrix})], the entries of the vector $z^{ij}(k)$ are
given by \begin{equation}[z^{ij}(k)]_l = \left\{
\begin{array}{cc}
\Big({1\over 2}-\epsilon\Big) (z_j(k)-z_i(k)) & \text{if } l=i, \\
{1\over 2} (z_j(k)-z_i(k)) & \text{if } l=j,\\
0& \text{otherwise.}
\end{array}%
\right. \label{intvector}\end{equation} This implies that the vector
norm $\|z^{ij}(k)\|_\infty$ can be upper-bounded by
\[\|z^{ij}(k)\|_{\infty}  \le  {1\over 2} \Big[\max_{l\in {\cal N}}\, z_l(k)
- \min_{l\in {\cal N}} z_l(k)\Big] \qquad\hbox{for all }i,j, \hbox{
and }k\ge 0.\] Defining $M(k)= \max_{l\in {\cal N}}\, z_l(k)$ and
$m(k)=\min_{l\in {\cal N}}\, z_l(k)$ for all $k\ge 0$, this implies
that
\[\|z^{ij}(k)\|_{\infty}  \le  {1\over 2} (M(k)-m(k))\le {1\over 2} \delta^k\, (M(0)-m(0))
\qquad\hbox{for all }i,j, \hbox{ and }k\ge 0,\] where the second
inequality follows from Lemma \ref{dslim}(b) in Appendix C.
Combining the preceding relation with Eq.\ (\ref{infbd}), we obtain
\[\|DT^k z(0)\|_{\infty} \le  {1\over 2n} \left(\sum_{i,j} p_{ij} \alpha_{ij}\right) \delta^k\,
(M(0)-m(0)).\] By Eq.\ (\ref{fundmat}), it follows that
\begin{eqnarray*}\|DYz(0)\|_\infty \le \sum_{k=0}^\infty \|D T^kz(0)\|_\infty
\le \sum_{k=0}^\infty {1\over 2n} \left(\sum_{i,j} p_{ij}
\alpha_{ij}\right) \delta^k\, (M(0)-m(0)) \le{\sum_{i,j} p_{ij}
\alpha_{ij}\over 2n (1-\delta)},
 \end{eqnarray*}
where to get the last inequality, we used the fact that $0\le
\delta<1$ and $M(0)-m(0)\le 1$, which follows from
$\|z(0)\|_{\infty}=1$. Since $z(0)$ is an arbitrary vector with
$\|z(0)\|_{\infty}=1$, this implies that
\[\|DY\|_\infty =\min_{\{z\ |\ \|z\|_{\infty}=1\}}\|DYz\|_\infty \le {1\over 2n (1-\delta)}
\left(\sum_{i,j} p_{ij} \alpha_{ij}\right). \] Combining this bound
with Eq.\ (\ref{normdif}), we obtain
\[\Big\|\bar{\pi}-{1\over n} e\Big\|_{\infty} \le {1\over 1-\delta} \, {\sum_{i,j} p_{ij}
\alpha_{ij}\over 2n},\] establishing the desired relation.

\vskip .5pc

\noindent (b) \ By Lemma \ref{tildewlim}(b), we have
\[E[\bar{x}] = \bar{\pi}'x(0).\]
This implies that
\[\Big|E[\bar{x}] - {1\over n} \sum_{i=1}^n x_i(0)\Big| =
\Big|\bar{\pi}'x(0)- {1\over n} e'x (0)\Big| \le
\Big\|\bar{\pi}-{1\over n} e\Big\|_{\infty} \|x(0)\|_{\infty}.\] The
result follows by combining this relation with part (a).
\end{proof}

Before providing the intuition for the preceding theorem, we provide a
related bound on the $l_{2}$-norm of the difference between $\bar{\pi}$ and
the uniform distribution ${\frac{1}{n}}e$ in terms of the second largest
eigenvalue of the social network matrix $T$, and then return to the
intuition for both results.

\begin{theorem}
\label{l2bd_thm} Let $\bar{\pi}$ denote the consensus distribution (cf.\
Lemma \ref{tildewlim}). The $l_2$-norm of the difference between $\bar{\pi}$
and ${\frac{1}{n}}e$ is given by
\begin{equation*}
\Big\|\bar{\pi}-{\frac{1}{n}} e\Big\|_{2} \le {\frac{1}{1-\lambda_2(T)}} \ {%
\frac{\sum_{i,j} p_{ij}\alpha_{ij}}{n}},
\end{equation*}
where $\lambda_2(T)$ is the second largest eigenvalue of the matrix $T$
defined in Eq.\ (\ref{decomposition}).
\end{theorem}

\begin{proof} Following a similar argument as in the proof of
Theorem \ref{delta_bd_thm}, we obtain
\begin{equation}
\Big\|\bar{\pi}- {1\over n}e\Big\|_{2} \le \|DY\|_{2},\label{l2dif}
\end{equation}
where $\|DY\|_2$ is the matrix norm induced by the $l_{2}$ vector
norm. To obtain an upper bound on the matrix norm $\|DY\|_2$, we
 consider an initial vector $z(0)\in \mathbb{R}^n$ with $\|z(0)\|_2=1$
and the sequence generated by
\[z(k+1) = T z(k)\qquad \hbox{for all } k\ge 0.\]
Then, for all $k\ge 0$, we have \begin{equation}DT^k z(0) = {1\over
n} \sum_{i,j} p_{ij} \alpha_{ij}
z^{ij}(k),\label{sumcomp}\end{equation} where the entries of the
vector $z^{ij}(k)$ are given by Eq.\ (\ref{intvector}). We can
provide an upper bound on the $\|z^{ij}(k)\|_2^2$ as
\[\|z^{ij}(k)\|_2^2 = {1\over 2} (z_j(k)-z_i(k))^2  = {1\over 2} \Big((z_j(k)-\bar z)+(\bar z- z_i(k))\Big)^2, \]
where $\bar z={1\over n} \sum_{l=1}^n z_l(k)$ for all $k$ (note that
since $T$ is a doubly stochastic matrix, the average of the entries
of the vector $z(k)$ is the same for all $k$). Using the relation
$(a+b)^2\le 2 (a^2 + b^2)$ for any scalars $a$ and $b$, this yields
\begin{equation}\|z^{ij}(k)\|_2^2 \le \sum_{l=1}^n (z_l(k)-\bar z)^2 =
\|z(k)-\bar ze\|_2^2.\label{l2bdvec}\end{equation} We have
\[z(k+1)-\bar z e = Tz(k) -\bar z e = T\Big(z(k)-\bar z e\Big),\]
where the second equality follows from the stochasticity of the
matrix $T$, implying that $Te=e$. The vector $z(k)-\bar z e$ is
orthogonal to the vector $e$, which is the eigenvector corresponding
to the largest eigenvalue $\lambda_1=1$ of matrix $T$ (note that
$\lambda_1=1$ since $T$ is a primitive and stochastic matrix).
Hence, using the variational characterization of eigenvalues, we
obtain
\[\|z(k+1)-\bar ze\|_2^2\le (z(k)-\bar z e)'\, T^2\,  (z(k)-\bar z e) \le \lambda_2(T)^2 \|z(k)-\bar ze\|_2^2.\]
where $\lambda_2(T)$ is the second largest eigenvalue of matrix $T$,
which implies
\[\|z(k)-\bar ze\|_2^2 \le \Big(\lambda_2(T)^2\Big)^k \|z(0)-\bar ze\|_2^2\le \lambda_2(T)^{2k}.\]
Here the second inequality follows form the fact that $\|z(0)\|_2=1$
and $\bar z$ is the average of the entries of vector $z(0)$.
Combining the preceding relation with Eq.\ (\ref{l2bdvec}), we
obtain
\[\|z^{ij}(k)\|_2\le \lambda_2(T)^k\qquad \hbox{for all }k\ge 0.\]
By Eq.\ (\ref{sumcomp}), this implies that
\[\|DT^k z(0)\|_2 = {1\over
n} \Big(\sum_{i,j} p_{ij} \alpha_{ij}\Big)\lambda_2(T)^k \qquad
\hbox{for all }k\ge 0. \] Using the definition of the fundamental
matrix $Y$, we obtain
\[\|DYz(0)\|_2 \le \sum_{k=0}^\infty \|D T^kz(0)\|_2
\le \sum_{k=0}^\infty {1\over n} \Big(\sum_{i,j} p_{ij}
\alpha_{ij}\Big)\lambda_2(T)^k = {1\over 1-\lambda_2(T)}\
{\sum_{i,j} p_{ij} \alpha_{ij}\over n},
\]
for any vector $z(0)$ with $\|z(0)\|_2=1$. Combined with Eq.\
(\ref{l2dif}), this yields the desired result.
\end{proof}

Theorem \ref{l2bd_thm} characterizes the variation of the stationary
distribution in terms of the average influence, ${\frac{\sum_{i,j}p_{ij}%
\alpha _{ij}}{n}}$, and the second largest eigenvalue of the social network
matrix $T$, $\lambda _{2}(T)$. As is well known, the difference $1-\lambda
_{2}(T)$, also referred to as the \textit{spectral gap,} governs the rate of
convergence of the Markov Chain induced by the social network matrix $T$ to
its stationary distribution (see \cite{bremaud}). In particular, the larger $%
1-\lambda _{2}(T)$ is, the faster the $k^{th}$ power of the transition
probability matrix converges to the stationary distribution matrix (cf.\
Theorem \ref{MC}). When the Markov chain converges to its stationary
distribution rapidly, we say that the Markov chain is \textit{fast-mixing}.%
\footnote{%
We use the terms \textquotedblleft spectral gap of the Markov chain" and
\textquotedblleft spectral gap of the (induced) graph", and
\textquotedblleft fast-mixing Markov chain" and \textquotedblleft
fast-mixing graph" interchangeably in the sequel.}

In this light, Theorem \ref{l2bd_thm} shows that, in a fast-mixing
graph, given a fixed average influence
${\frac{\sum_{i,j}p_{ij}\alpha _{ij}}{n}}$, the consensus
distribution is \textquotedblleft closer\textquotedblright\ to the
underlying $\theta ={\frac{1}{n}}\sum_{i=1}^{n}x_{i}(0)$ and the
extent of misinformation is limited. This is intuitive. In a
fast-mixing social network graph, there are several connections
between any pair of agents. Now for any forceful agent, consider
the set of agents who will have some influence on his beliefs.
This set itself is connected to the rest of the agents and thus
obtains information from the rest of the society. Therefore, in a
fast-mixing graph (or in a society represented by such a graph),
the beliefs of forceful agents will themselves be moderated by the
rest of the society before they spread widely. In contrast, in a
slowly-mixing graph, we can have a high degree of clustering
around forceful agents, so that forceful agents get their (already
limited) information intake mostly from the same agents that they
have influenced. If so, there will be only a very indirect
connection from the rest of the society to the beliefs of forceful
agents and forceful agents will spread their information widely
before their opinions also adjust. As a result, the consensus is
more likely to be much closer to the opinions of
forceful agents, potentially quite different from the true underlying state $%
\theta $.

This discussion also gives intuition for Theorem \ref{delta_bd_thm} since the constant $%
\delta $ in that result is closely linked to the mixing properties of the
social network matrix and the social network graph. In particular, Theorem %
\ref{delta_bd_thm} clarifies that $\delta $ is related to the maximum
shortest path and the minimum probability of (indirect) communication
between any two agents in the society. These two notions also crucially
influence the spectral gap $1-\lambda _{2}(T_{n})$, which plays the key role
in Theorem \ref{l2bd_thm}.

These intuitions are illustrated in the next example, which shows how in a
certain class of graphs, misinformation becomes arbitrarily small as the
social network grows.

\begin{example}
\textbf{(Expander Graphs)} Consider a sequence of social network graphs $%
\mathcal{G}_n = (\mathcal{N}_n,\mathcal{A}_n)$ induced by symmetric $n\times
n$ matrices $T_n$ [cf.\ Eq.\ (\ref{undirgraph})]. Assume that this sequence
of graphs is a \emph{family of expander graphs}, i.e., there exists a
positive constant $\gamma>0$ such that the spectral gap $1-\lambda_2(T_n)$
of the graph is uniformly bounded away from 0, independent of the number of
nodes $n$ in the graph, i.e.,
\begin{equation*}
\gamma \le 1-\lambda_2(T_n)\qquad \hbox{for all }n,
\end{equation*}
(see \cite{fanchung}) As an example, Internet has been shown to be an
expander graph under the preferential connectivity random graph model (see
\cite{MPS} and \cite{Jacksonbook}). Expander graphs have high connectivity
properties and are fast mixing.

We consider the following influence structure superimposed on the social
network graph $\mathcal{G}_{n}$. We define an agent $j$ to be \textit{%
locally forceful} if he influences a constant number of agents in
the society, i.e., his total influence, given by
$\sum_{i}p_{ij}\alpha _{ij}$, is a constant independent of $n$. We
assume that there is a constant number of locally forceful agents.
Let $\bar{\pi}_{n}$ denote the stationary distribution of the Markov
Chain with transition probability matrix given by the mean
interaction matrix $\tilde{W}$ [cf.\ Eq.\ (\ref{meanint})]. Then, it
follows from Theorem \ref{l2bd_thm} that
\begin{equation*}
\Big\|\bar{\pi}_{n}-{\frac{1}{n}}e\Big\|_{2}\rightarrow 0\quad \hbox{as}%
\quad n\rightarrow \infty .
\end{equation*}%
This shows that if the social network graph is fast-mixing and there
is a constant number of locally forceful agents, then the difference
between the consensus belief and the average of the initial beliefs
vanishes. Intuitively, in expander graphs, as $n$ grows large, the
set of individuals who are the source of information of forceful
agents become highly connected, and thus rapidly inherit the average
of the information of the rest of the society. Provided that the
number of forceful agents and the impact of each forceful agent do
not grow with $n$, then their influence becomes arbitrarily small as
$n$ increases.
\end{example}


\section{Connectivity of Forceful Agents and Misinformation\label%
{local-misinfo}}

The results provided so far exploit the decomposition of the evolution of
beliefs into the social network component (matrix $T$) and the influence
component (matrix $D$). This decomposition does not exploit the interactions
between the structure of the social network and the location of forceful
agents within it. For example, forceful agents located in different parts of
the same social network will have different impacts on the extent of
misinformation in the society, but our results so far do not capture this
aspect. The following example illustrates these issues in a sharp way.

\begin{example}
\label{bottleneck_ex}
\begin{figure}[tbp]
\centering
\includegraphics[width=.8\textwidth]{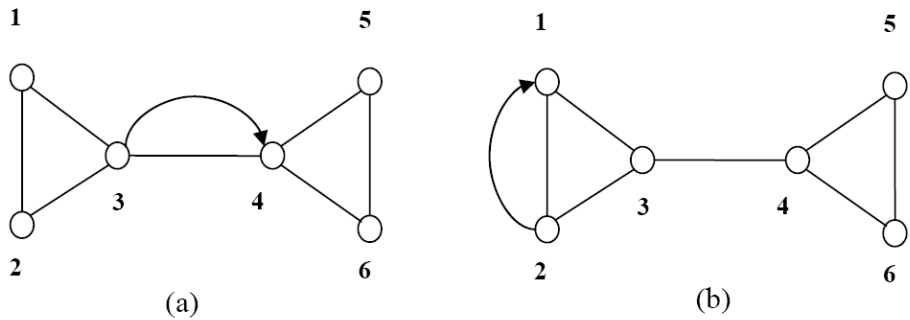}\newline
\caption{Impact of location of forceful agents on the stationary
distribution (a) Misinformation over the bottleneck (b) Misinformation
inside a cluster}
\label{cluster_fig}
\end{figure}

Consider a society consisting of six agents and represented by the
(undirected) social network graph shown in Figure \ref{cluster_fig}. The
weight of each edge $\{i,j\}$ is given by
\begin{equation*}
\lbrack T]_{ij}={\frac{1}{2n}}(p_{ij}+p_{ji}),
\end{equation*}%
where, for illustration, we choose $p_{ij}$ to be inversely
proportional to the degree of node $i$, for all $j$. The self-loops
are not shown in Figure \ref{cluster_fig}.

We distinguish two different cases as illustrated in Figure \ref{cluster_fig}%
. In each case, there is a single forceful agent and $\alpha =1/2$. This is
represented by a directed forceful link. The two cases differ by the
location of the forceful link, i.e., the forceful link is over the
bottleneck of the connectivity graph in part (a) and inside the left cluster
in part (b). The corresponding consensus distributions can be computed as
\begin{equation*}
\pi _{a}=\frac{1}{6}(1.25,1.25,1.25,0.75,0.75,0.75)^{\prime },\quad \pi _{b}=%
\frac{1}{6}(0.82,1.18,1,1,1,1)^{\prime }.
\end{equation*}%
Even though the social network matrix $T$ (and the corresponding graph) is
the same in both cases, the consensus distributions are different. In
particular, in part (a), each agent in the left cluster has a higher weight
compared to the agents in the right cluster, while in part (b), the weight
of all agents, except for the forceful and influenced agents, are equal and
given by 1/6. This is intuitive since when the forceful link is over a
bottleneck, the misinformation of a forceful agent can spread and influence
a larger portion of the society before his opinions can be moderated by the
opinions of the other agents.
\end{example}

This example shows how the extent of spread of misinformation varies
depending on the location of the forceful agent. The rest of this section
provides a more detailed analysis of how the location and connectivity of
forceful agents affect the formation of opinions in the network. We proceed
as follows. First, we provide an alternative exact characterization of
excess influence using mean first passage times. We then\ introduce the
concept of \emph{essential edges}, similar to the\ situation depicted in
Example \ref{bottleneck_ex}, and provide sharper exact results for graphs in
which forceful links coincide with essential edges. We then generalize these
notions to more general networks by introducing the concept of \emph{%
information bottlenecks}, and finally, we develop new techniques for
determining tighter upper bounds on excess influence by using ideas from
graph clustering.

\subsection{Characterization in Terms of Mean First Passage Times}

Our next main result provides an exact characterization of the excess
influence of agent $i$ in terms of the mean passage times of the Markov
chain with transition probability matrix $T$. This result, and those that
follow later in this section, will be useful both to provide more
informative bounds on the extent of misinformation and also to highlight the
sources of excess influence for certain agents in the society.

We start with presenting some basic definitions and relations (see Chapter 2
of \cite{aldous}).

\begin{definition}
Let $(X_t, t =0, 1, 2, \ldots)$ denote a discrete-time Markov chain. We
denote the \emph{first hitting time} of state $i$ by
\begin{equation*}
T_i = \inf\, \{t\geq 0\ |\ X_t = i\}.
\end{equation*}
We define the \emph{mean first passage time} from state $i$ to state $j$ as
\begin{equation*}
m_{ij} = \mathbb{E}[T_j\ |\ X_0 = i],
\end{equation*}
and the \emph{mean commute time} between state $i$ and state $j$ as $%
m_{ij}+m_{ji}$. Moreover, we define the \emph{mean first return time} to a
particular state $i$ as
\begin{equation*}
m_{i}^+ = \mathbb{E}[T_i^+\ |\ X_0 = i],
\end{equation*}
where
\begin{equation*}
T_i^+ = \inf\, \{t\geq 1\ |\ X_t = i\}.
\end{equation*}
\end{definition}

\begin{lemma}
\label{facts_lemma} \emph{Consider a Markov chain with transition
matrix $Z$ and stationary distribution $\pi$. We have:
\begin{itemize}
  \item[(i)] The mean first return time from state $i$ to $i$ is given by $m_i^+ = 1/{\pi_i}.$
  \item[(ii)] The mean first passage time from $i$ to $j$ is given by
  $$m_{ij} = \frac{Y_{jj}-Y_{ij}}{\pi_j},$$
  where $Y =\sum_{k=0}^\infty (Z^k-Z^{\infty})$ is the fundamental matrix of the Markov chain.
\end{itemize}}
\end{lemma}

We use the relations in the preceding lemma between the fundamental matrix
of a Markov chain and the mean first passage times between states, to
provide an exact characterization of the excess influence of agent $k$.

\begin{theorem}
\label{sensitivity_thm} Let $\bar{\pi}$ denote the consensus distribution.
We have:
\begin{itemize}
\item[(a)] For every agent $k$
\begin{equation*}  \label{exact_sensitivity}
\bar \pi_k - \frac{1}{n} = \frac{1}{2n^2}\sum_{i,j} p_{ij}\alpha_{ij}\Big( %
(1-2\epsilon)\bar \pi_i +\bar \pi_j\Big)\big(m_{ik} - m_{jk}\big).
\end{equation*}
\item[(b)] Let $\mathcal{A_{I}}$ denote the set of edges over which there is a forceful link,
i.e.,
\begin{equation*}
\mathcal{A_{I}}=\Big\{\{i,j\}\in \mathcal{A}\ |\ \alpha _{ij}>0\hbox{ or }%
\alpha _{ji}>0\Big\}.
\end{equation*}%
Assume that for any $\{i,j\},\{k,l\}\in \mathcal{A_{I}}$, we have $%
\{i,j\}\cap \{k,l\}=\emptyset $. Then,
\begin{equation}
\bar{\pi}_{k}-\frac{1}{n}=\frac{1}{n^{3}}\sum_{i,j}\frac{p_{ij}\alpha
_{ij}(1-\epsilon )}{1-\zeta _{ij}/n^{2}}(m_{ik}-m_{jk}),
\label{exact_sensitivity_disjoint}
\end{equation}
\end{itemize}
where
\begin{equation*}
\zeta_{ij} = \zeta_{ji}= \Big[\big({\frac{1}{2}}+\epsilon\big)p_{ij}
\alpha_{ij} - {\frac{1}{2}} p_{ji} \alpha_{ji}\Big]m_{ij} - \Big[{\frac{1}{2}%
} p_{ij} \alpha_{ij} - \big({\frac{1}{2}}+\epsilon\big) p_{ji} \alpha_{ji}%
\Big]m_{ji},
\end{equation*}
and $m_{ij}$ is the mean first passage time from state $i$ to state
$j$ of a Markov chain with transition matrix given by the social
network matrix $T$ [cf.\ Eq.\ (\ref{decomposition})].
\end{theorem}

\begin{proof} (a) \ Following the same line of argument as in the proof of Theorem \ref{delta_bd_thm}, we
can use the perturbation results of Theorem \ref{pertbd} to write
the excess influence of agent $k$ as
\begin{equation}\label{pertbd_eq}
\bar \pi_k - \frac{1}{n} = \bar \pi' D [Y]^k,
\end{equation}
where $Y$ is the fundamental matrix of a Markov chain with transition matrix $T$. Using
(\ref{updatematrix}), and the definition of $D$ in (\ref{decomposition}) we have
\[\Big[D[Y]^k\Big]_l = \sum_{i,j} \frac{p_{ij}\alpha_{ij}}{n} \left\{
\begin{array}{cc}
\left({1\over 2} - \epsilon \right) (Y_{jk}-Y_{ik}) & \text{if } l=i, \\
{1\over 2}(Y_{jk}-Y_{ik}) & \text{if } l=j,\\
0& \text{otherwise.}
\end{array}%
\right. \]
Hence, we can write right-hand side of Eq.\
(\ref{pertbd_eq}) as follows:
\begin{equation}\label{pertbd_eq_k}
\bar \pi_k - \frac{1}{n} = \sum_{i,j} \frac{p_{ij}\alpha_{ij}}{2n}\big( (1-2\epsilon)\bar \pi_i
+\bar \pi_j\big)\big(Y_{jk} - Y_{ik}\big).
\end{equation}
By Lemma \ref{facts_lemma}(ii), we have
\begin{equation}\label{mpt_fund}
Y_{jk} - Y_{ik}  = \frac{1}{n}(m_{ik} - m_{jk}),
\end{equation}
where $Y$ is the fundamental matrix of the Markov chain with
transition matrix $T$. The desired result follows by substituting
the preceding relation in Eq.\ (\ref{pertbd_eq_k}).

\vskip .5pc

\noindent (b) \  In view of the assumption that all edges in $\cal A_I$ are pairwise disjoint, the
perturbation matrix $D$ decomposes into disjoint blocks, i.e.,
\begin{equation}
D = \sum_{\{i,j\} \in \cal A_I} D_{ij}+D_{ji}, \quad \textrm{where } D_{ij} =\frac{ p_{ij}
\alpha_{ij} }{n}\big[J_{ij}-A_{ij}\big].\label{disjoint-inf}\end{equation} For each edge $\{i,j\}
\in \cal A_I$, it is straightforward to show that
$$\Big((D_{ij}+D_{ji})Y\Big)^2  = \Big(1- \frac{\zeta_{ij}}{n^2}\Big)(D_{ij}+D_{ji})Y.$$
Using the decomposition in Eq.\ (\ref{disjoint-inf}) and the preceding relation, it can be seen
that
$$DY(I-DY)^{-1} = \sum_{i,j}\Big(1- \frac{\zeta_{ij}}{n^2}\Big)^{-1}
D_{ij}Y.$$ Combined with the exact perturbation result in Theorem \ref{pertbd}, this implies that
\begin{eqnarray}\label{pert_disjoint_eq}
\bar \pi_k - {1 \over n} &=& {1 \over n}[e'DY(I-DY)^{-1}]_k \nonumber \\
                          &=& {1 \over n}\sum_{i,j} \Big(1- \frac{\zeta_{ij}}{n^2}\Big)^{-1}
                          [e'D_{ij}Y]_k \nonumber \\
                          & = & \sum_{i,j} \frac{\frac{ p_{ij} \alpha_{ij} }{n^2} (1-\epsilon)}{1- \zeta_{ij}/{n^2}} (Y_{jk} -
                          Y_{ik}). \nonumber
\end{eqnarray}
The main result follows by substituting Eq.\ (\ref{mpt_fund}) in the above equation.
\end{proof}

Part (a) of Theorem \ref{sensitivity_thm} provides an exact expression for
the excess influence of agent $k$ as a function of the mean first passage
times from agent (or state) $k$ to the forceful and influenced agents. The
excess influence of each agent therefore depends on the relative \emph{%
distance} of that agent to the forceful and the influenced agent. To
provide an intuition for this result, let us consider the special
case in which there is a single forceful link $(j,i)$ in the society
(i.e., only one pair of agents $i$ and $j$ with $\alpha _{ij}>0$)
and thus a single forceful agent $j $. Then for any agent $k$, their
only source of excess influence can come from their (potentially
indirect) impact on the beliefs of the forceful agent $j$. This is
why $m_{jk}$, which, loosely speaking, measures the distance between
$j$ and $k$, enters negatively in to the expression for the excess
influence of agent $k$. In addition, any agent who meets
(communicates) with agent $i$ with a high probability will be
indirectly influenced by the opinions of the forceful agent $j$.
Therefore, the excess influence of agent $k$ is increasing in his
distance to $i$, thus in $m_{ik}$. In particular, in the extreme
case where $m_{ik}$ is small, agent $k$ will have negative excess
influence (because he is very close to the heavily \textquotedblleft
influenced\textquotedblright\ agent $i$) and in the polar extreme,
where $m_{jk}$ is small, he will have positive excessive influence
(because his views will be quickly heard by the forceful agent $j$).
The general expression in part (a) of the theorem simply generalizes
this reasoning to general social networks with multiple forceful
agents and several forceful links.

Part (b) provides an alternative expression [cf.\ Eq.\
(\ref{exact_sensitivity_disjoint})], with a similar intuition for
the special case in which all forceful links are disjoint. The main
advantage of the expression in part (b) is that, though more
complicated, is not in terms of the expected consensus distribution $\bar{\pi%
}$ (which is endogenous). Disjoint forceful link property in part (b) is
also useful because it enables us to isolate the effects of the forceful
agents. The parameter $\zeta _{ij}$ in Eq.\ (\ref{exact_sensitivity_disjoint}%
) captures the asymmetry between the locations of agents $i$ and $j$ in the
underlying social network graph. Although the expression for excess influence in part (a) of Theorem \ref%
{sensitivity_thm} is a function of the consensus distribution
$\bar{\pi}$, each element of this vector (distribution) can be
bounded by 1 to obtain an upper bound for the excess influence of
agent $k$.

Using the results in Theorem \ref{sensitivity_thm}, the difference between
the consensus distributions discussed in Example \ref{bottleneck_ex} can be
explained as follows. In Example \ref{bottleneck_ex}(a), the mean first
passage time from agent 4 to any agent $k$ in the left cluster is strictly
larger than that of agent 3 to agent $k$, because every path from agent 4 to
the left cluster should pass through agent 3. Therefore, $m_{4k}>m_{3k}$ for
$k=1,2,3$, and agents in the left cluster have a higher consensus weight. In
Example \ref{bottleneck_ex}(b), due to the symmetry of the social network
graph, the mean first passage times of agents 1 and 2 to any agent $k\neq 1,2
$ are the same, hence establishing by Theorem \ref{sensitivity_thm} the
uniform weights in the consensus distribution.

In the following we study the effect of the location of a forceful link on
the excess influence of each agent by characterizing the relative mean first
passage time $|m_{ik}-m_{jk}|$, in terms of the properties of the social network graph.


\subsection{Forceful Essential Edges\label{ess_sec}}

In this subsection, we provide an exact characterization of the excess
influence of agent $k$ explicitly in terms of the properties of the social
network graph. We focus on the special case when the undirected edge between
the forceful and the influenced agent is \emph{essential} for the social
network graph, in the sense that without this edge the graph would be
disconnected. We refer to such edges as \emph{forceful essential edges}.
Graphs with forceful essential edges approximate situations in which a
forceful agent, for example a media outlet or political leader, itself
obtains all of its information from a tightknit community.

We first give the definition of an essential edge of an undirected graph.

\begin{definition}
\label{ess_def} Let $\mathcal{G}=(\mathcal{N,A})$ be an undirected graph. An
edge $\{i,j\}\in \mathcal{A}$ is an \emph{essential} edge of the graph $%
\mathcal{G}=(\mathcal{N,A})$ if its removal would partition the set of nodes
into two disjoint sets $N(i,j) \subset\mathcal{N}$ with $i\in N(i,j)$, and $%
N(j,i) \subset \mathcal{N}$ with $j\in N(j,i)$.
\end{definition}

The following lemma provides an exact characterization of the mean first
passage time from state $i$ to state $j$, where $i$ and $j$ are the end
nodes of an essential edge $\{i,j\}$.

\begin{lemma}
\label{ess_lemma} \emph{Consider a Markov chain with a doubly
stochastic transition probability matrix $T$. Let $\{i,j\}$ be an
essential edge of the social network graph induced by matrix $T$.
\begin{itemize}
  \item[(a)] We have $$m_{ij} = \frac{|N(i,j)|}{T_{ij}}.$$
  \item[(b)] For every $k \in N(j,i)$,
   $$m_{ik} - m_{jk} = m_{ij}. $$
\end{itemize}
}
\end{lemma}

\begin{proof}
Consider a Markov chain over the set of states $\mathcal N' = N(i,j) \cup \{j\}$, with transition
probabilities
\begin{equation}\label{sub_chain}
   \hat T_{kl} = T_{kl}, \quad \textrm{for all } k \neq l. \nonumber
\end{equation}

For the new chain with stationary distribution $\hat \pi$ we have
$$\hat \pi_j = \frac{T_{ij}}{\hat T} = \frac{T_{ij}}{|N(i,j)| + T_{ij}},$$
where $\hat T$ is the total edge weight in the new chain.

Since $\{i,j\}$ is essential, every path from $i$ to $j$ should pass through $\{i,j\}$. Moreover,
because of equivalent transition probabilities between the new Markov chain and the original one on
$\mathcal N'$, the mean passage time $m_{ij}$ of the original Markov chain is equal to mean passage
time $\hat m_{ij}$ of the new chain. On the other hand, for the new chain, we can write the mean
return time to $j$ as
$${ \hat m_j^+} = 1+\hat m_{ij} = 1+m_{ij},$$
which implies [cf.\ Lemma \ref{facts_lemma}(i)]
$$m_{ij} = \frac{1}{\hat \pi_j} - 1 =  \frac{|N(i,j)|}{T_{ij}}.$$

The second part of the claim follows from the fact that all of the paths from $i$ to $k$ must pass
through $\{i,j\}$, because it is the only edge connecting $N(i,j)$ to $N(j,i)$. Thus, we conclude
$$m_{ik} = m_{ij} + m_{jk}.$$
\end{proof}

We use the relation in Lemma \ref{ess_lemma} to study the effect of a single
forceful link over an essential edge on the excess influence of each agent.

\begin{theorem}
\label{ess_thm} Let $\bar{\pi}$ denote the consensus distribution. Assume
that there exists a single pair $\{i,j\}$ for which the influence
probability $\alpha _{ij}>0$. Assume that the edge $\{i,j\}$ is an essential
edge of the social network graph. Then, we have for all $k$,
\begin{equation}
\bar{\pi}_{k}-{\frac{1}{n}}=\frac{2}{n^{2}}\frac{\theta _{ij}(1-\epsilon )}{%
1-{\frac{\theta _{ij}}{n}}\Big((1+2\epsilon )|N(i,j)|-|N(j,i)|\Big)}\, \Psi _{ij}(k), \nonumber
\label{ess_sensitivity}
\end{equation}%
where
\begin{equation*}
\theta _{ij}=\frac{p_{ij}\alpha _{ij}}{p_{ij}(1-\gamma
_{ij})+p_{ji}(1-\gamma _{ji})},
\end{equation*}%
and
\begin{equation*}
\Psi _{ij}(k)=\left\{
\begin{array}{ll}
|N(i,j)|, & \hbox{$k \in N(j,i)$,} \\
-|N(j,i)|, & \hbox{$k \in N(i,j)$.}%
\end{array}%
\right.
\end{equation*}
\end{theorem}

\begin{proof}
Since edge $\{i,j\}$ is essential, by Lemma \ref{ess_lemma} we have for every $k \in N(j,i)$
\begin{equation}\label{ess_passage}
m_{ik} - m_{jk} = m_{ij} = \frac{|N(i,j)|}{T_{ij}} = \frac{2n|N(i,j)|}{p_{ij}(1-\gamma_{ij}) +
p_{ji}(1-\gamma_{ji})}, \nonumber
\end{equation}
Similarly, for every $k \in N(i,j)$, we obtain
$$m_{ik} - m_{jk} = -m_{ji} =  -\frac{2n|N(j,i)|}{p_{ij}(1-\gamma_{ij}) + p_{ji}(1-\gamma_{ji})}.$$
Combining the preceding relations, we can write for the relative mean passage time $m_{ik} - m_{jk}
= \frac{2n\theta_{ij}}{p_{ij}\alpha_{ij}} \Psi_{ij}(k)$. Since $(i,j)$ is the only forceful link,
we can apply  Theorem \ref{sensitivity_thm}(b) to get
$$  \bar \pi_k - {1\over n} = \bigg(\frac{2}{n^2}\bigg) \frac{\theta_{ij}(1-\epsilon)}{1 - \zeta_{ij}/n^2 } \cdot  \Psi_{ij}(k),$$
where $\zeta_{ij}$ is given by
$$\zeta_{ij}= \frac{p_{ij} \alpha_{ij}}{2} \big[(1+2\epsilon)m_{ij} - m_{ji}\big].$$
Combining the above relations with  Lemma \ref{ess_lemma}(i) establishes the desired result.
\end{proof}

Theorem \ref{ess_thm} shows that if two clusters of agents, e.g.,
two communities, are connected via an essential edge over which
there is a forceful link, then the excess influence of all agents
within the same cluster are equal (even when the cluster does not
have any symmetry properties). This implies that the opinions of all
agents that are in the same cluster as the forceful agent affect the
consensus opinion of the society with the same strength. This
property is observed in part (a) of Example \ref{bottleneck_ex}, in
which edge \{3,4\} is an essential edge. Intuitively, all of the
agents in that cluster will ultimately shape the opinions of the
forceful agent and this is the source of their excess influence. The
interesting and surprising feature is that they all have the same
excess influence, even if only some of them are directly connected
to the forceful agent. Loosely speaking, this can be explained using
the fact that, in the limiting distribution, it is the consensus
among this cluster of agents that will impact the beliefs of the
forceful agent, and since within this cluster there are no other
forceful agents, the consensus value
among them puts equal weight on each of them (recall Corollary \ref%
{corollary}).


\subsection{Information Bottlenecks}

We now extend the ideas in Theorem \ref{ess_thm} to more general
societies. We observed in Example \ref{bottleneck_ex} and Section \ref%
{ess_sec} that influence over an essential edge can have global
effects on the consensus distribution since essential edges are
``bottlenecks" of the information flow in the network. In this
subsection we generalize this idea to influential links over
bottlenecks that are not necessarily essential edges as defined in
Definition \ref{ess_def}. Our goal is to study the impact of
influential links over bottlenecks on the consensus distribution.

To achieve this goal, we return to the characterization in Theorem \ref%
{sensitivity_thm}, which was in terms of first mean passage times,
and then
provide a series of (successively tighter) upper bounds on the key term $%
(m_{ik}-m_{jk})$ in Eq. (\ref{exact_sensitivity_disjoint}) in this
theorem. Our first bound on this object will be in terms of the
\emph{minimum normalized cut} of a Markov chain (induced by an
undirected weighted graph), which is introduced in the next
definition. We will use the term \textit{cut of a Markov Chain}
(or \textit{cut of an undirected graph}) to denote a partition of
the set of states of a Markov chain (or equivalently the nodes of
the corresponding graph) into two sets.

\begin{definition}
\label{cond_def} Consider a Markov chain with set of states
$\mathcal{N}$,
symmetric transition probability matrix $Z$, and stationary distribution $%
\pi $. The\emph{\ minimum normalized cut value} (or
\emph{conductance}) of the Markov chain, denoted by $\rho$, is
defined as
\begin{equation}  \label{conductance}
\rho = \inf_{S \subset \mathcal{N}} \frac{Q(S,S^c)}{\pi(S) \pi(S^c)},
\end{equation}
where $Q(A,B) = \sum_{i \in A, j \in B} \pi_i Z_{ij}$, and $\pi(S)
= \sum_{i \in S} \pi_i$. We refer to the cut that achieves the
minimum in this optimization problem as the \textit{minimum
normalized cut}.
\end{definition}

The objective in the optimization problem in (\ref{conductance})
is the (normalized) conditional probability that the Markov chain
makes a transition from a state in set $S$ to a state in set
$S^{c}$ given that the initial state is in $S$. The minimum
normalized cut therefore characterizes how fast the Markov chain
will escape from any part of the state space, hence is an
appropriate measure of information bottlenecks or the mixing time
of the underlying graph. Clearly, the minimum normalized cut value
is larger in more connected graphs.

The next lemma provides a relation between the maximum mean
commute time of a Markov chain (induced by an undirected graph)
and the minimum normalized
cut of the chain, which is presented in Section 5.3 of Aldous and Fill \cite%
{aldous}. This result will then be used in the next theorem to
provide an improved bound on the excess influences by using the
fact that $\left\vert m_{ik}-m_{jk}\right\vert \leq
\max_{i,j}\left\{ m_{ij},m_{ji}\right\} $ (see, in particular,
proof of Theorem \ref{glob_cond_thm}).

\begin{lemma}
\label{cond_lemma} \emph{Consider an $n$-state Markov chain with transition
matrix $Z$ and stationary distribution $\pi$. Let $\rho$ denote the minimum
normalized cut value of the Markov chain (cf.\ Definition \ref{cond_def}).
The maximum mean commute time satisfies the following relation:
\begin{equation}  \label{commute-rel}
\max_{i,j}\{m_{ij}+m_{ji}\} \leq \frac{4(1+\log n)}{\rho \min {\pi_k}}.
\end{equation}%
}
\end{lemma}

We use the preceding relation together with our characterization of excess
influence in terms of mean first passage times in Theorem \ref%
{sensitivity_thm} to obtain a tighter upper bound on the $l_{\infty }$ norm
of excess influence than in Theorem \ref{delta_bd_thm}. This result, which
is stated next, both gives a more directly interpretable limit on the extent
of misinformation in the society and also shows the usefulness of the
characterization in terms of mean first passage times in Theorem \ref%
{sensitivity_thm}.

\begin{theorem}
\label{glob_cond_thm} Let $\bar{\pi}$ denote the consensus distribution. Then, we have
\begin{equation}
\Big\|\bar{\pi}-\frac{1}{n}e\Big\|_{\infty }\leq \sum_{i,j}\frac{%
2p_{ij}\alpha _{ij}}{n}\Big(\frac{1+\log n}{\rho }\Big), \nonumber \label{spectral_bound}
\end{equation}%
where $\rho $ is the minimum normalized cut value of the Markov chain with
transition probability matrix given by the social network matrix $T$ (cf.\ Definition \ref{cond_def}).
\end{theorem}

\begin{proof}
By Theorem \ref{sensitivity_thm} we have for every $k$
\begin{eqnarray}
  \Big| \bar \pi_k  - \frac{1}{n}\Big| &=& \sum_{i,j} \frac{p_{ij}\alpha_{ij}}{2n^2}( (1-2\epsilon)\bar \pi_i
+\bar \pi_j)\big|m_{ik} - m_{jk}\big| \nonumber \\
&\leq& \sum_{i,j} \frac{p_{ij}\alpha_{ij}}{2n^2}\big|m_{ik} - m_{jk}\big| \nonumber \\
&\leq& \sum_{i,j} \frac{p_{ij}\alpha_{ij}}{2n^2} \max\{m_{ij}, m_{ji}\} \label{sen_ieq_3} \\
&\leq& \sum_{i,j} \frac{p_{ij}\alpha_{ij}}{2n^2} \max_{i,j} \{m_{ij}+m_{ji}\} \nonumber \\
&\leq& \sum_{i,j} \frac{2p_{ij}\alpha_{ij}}{n} \Big(\frac{1+\log
n}{\rho}\Big), \nonumber
\end{eqnarray}
where (\ref{sen_ieq_3}) holds because $m_{ik} \leq m_{ij} + m_{jk}$,
and $m_{jk} \leq m_{ji} + m_{ik}$, and the last inequality follows
from Eq.\ (\ref{commute-rel}), and the fact that $\pi =
\frac{1}{n}e$.
\end{proof}

One advantage of the result in this theorem is that the bound is in
terms of $\rho $, the minimum normalized cut of the social network
graph. As emphasized in Definition \ref{cond_def}, this notion is
related to the strength of
(indirect) communication links in the society. Although the bound in Theorem %
\ref{glob_cond_thm} is tighter than the one we provided in Theorem \ref%
{delta_bd_thm}, it still leaves some local information unexploited
because it focuses on the maximum mean commute times between all
states of a Markov chain. The following example shows how this bound
may be improved further by focusing on the mean commute time between
the forceful and the influenced agents.

\begin{example}
\label{barbell_ex} \textbf{(Barbell graph)} The barbell graph consists of
two complete graphs each with $n_1$ nodes that are connected via a path that
consists of $n_2$ nodes (cf.\ Figure \ref{barbell_fig}). Consider the
asymptotic behavior
\begin{equation*}
n \rightarrow \infty, \quad n_1/n \rightarrow \nu, \quad n_2/n \rightarrow
1-2\nu,
\end{equation*}
where $n=2n_1+n_2$ denotes the total number of nodes in the barbell graph,
and $0 < \nu < \frac{1}{2}$. The mean first passage time from a particular
node in the left bell to a node in the right bell is $O(n^3)$ as $n
\rightarrow \infty$, while the mean passage time between any two nodes in
each bell is $O(n)$ (See Chapter 5 of \cite{aldous} for exact results).
Consider a situation where there is a single forceful link in the left bell.

The minimum normalized cut for this example is given by cut $C_{0}$,
with normalized cut value $O(1/n)$, which captures the bottleneck in
the global network structure. Since the only forceful agent is
within the left bell in this example, we expect the flow of
information to be limited by cuts that separate the forceful and the
influenced agent, and partition the left bell. Since the left bell
is a complete graph, the cuts associated with this part of the graph
will have higher normalized cut values, thus yielding tighter bounds
on the excess influence of the agents. In what follows, we consider
bounds in terms of \textquotedblleft relative
cuts\textquotedblright\ in the social network graph that separate
forceful and influenced agents in order to
capture bottlenecks in the spread of misinformation (for example, cuts $%
C_{1} $, $C_{2}$, and $C_{3}$ in Figure \ref{barbell_fig}).

\begin{figure}[tbp]
\centering
\includegraphics[width=.8\textwidth]{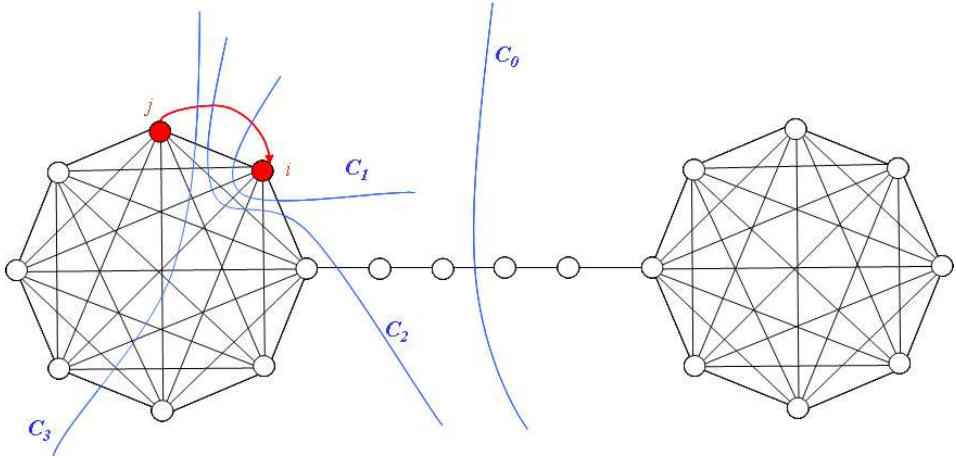}\newline
\caption{The barbell graph with $n_1 = 8$ nodes in each bell and $n_2 = 4$.
There is a single forceful link, represented by a directed link in the left
bell.}
\label{barbell_fig}
\end{figure}
\end{example}

\subsection{Relative Cuts}

The objective of this section is to improve our characterization of
the extent of misinformation in terms of information bottlenecks. To
achieve this objective, we introduce a new concept, \emph{relative
cuts}, and then show how this new concept is useful to derive
improved upper bounds on the excess influence of different
individuals and on the extent of misinformation. Our strategy is to
develop tighter bounds on the mean commute times between the
forceful and influenced agents in terms of relative cut values.
Together with Theorem \ref{sensitivity_thm}, this enables us to
provide bounds on the excess influence as a function of the
properties of the social network graph and the location of the
forceful agents within it.

\begin{definition}
\label{rel_cut_def} Let $\mathcal{G}=(\mathcal{N}, \mathcal{A})$ be
an undirected graph with edge $\{i,j\}$ weight given by $w_{ij}$.
The \textit{minimum relative cut value between $a$ and $b$}, denoted by $%
c_{ab}$, is defined as
\begin{equation*}
c_{ab} = \inf \bigg\{ \sum_{{{\{i,j\}\in \mathcal{A} \atop i \in S, j \in
S^c}}} w_{ij}\ |\ S \subset \mathcal{N}, a \in S, b \notin S \bigg\}.
\end{equation*}
We refer to the cut that achieves the minimum in this optimization problem
as the \textit{minimum relative cut}.
\end{definition}

The next theorem uses the extremal characterization of the mean
commute times presented in Appendix D, Lemma \ref{extrmal_commute},
to provide bounds on the mean commute times in terms of minimum
relative cut values.

\begin{theorem}
\label{mincut_thm} Let $\mathcal{G}=(\mathcal{N}, \mathcal{A})$ be the
social network graph induced by the social network matrix $T$ and
consider a Markov chain with transition matrix $T$. For any $a,b\in\mathcal{N%
}$, the mean commute time between $a$ and $b$ satisfies
\begin{equation}  \label{mincut_bound}
\frac{n}{c_{ab}} \leq m_{ab} + m_{ba} \leq \frac{n^2}{c_{ab}},
\end{equation}
where $c_{ab}$ is the minimum relative cut value between $a$ and $b$ (cf.\
Definition \ref{rel_cut_def}).
\end{theorem}

\begin{proof}
For the lower bound we exploit the extremal characterization of the
mean commute time given by Eq.\ (\ref{Dirichlet}) in Lemma
\ref{extrmal_commute}. For any $S \subset \mathcal N$ containing $a$
and not containing $b$, pick the function $g_S$ as follows:
$$g_S(i) = \left\{
    \begin{array}{ll}
      0, & \hbox{$i \in S$;} \\
      1, & \hbox{otherwise.}
    \end{array}
  \right.
$$
The function $g$ is feasible for the maximization problem in Eq.\
(\ref{Dirichlet}). Hence,
\begin{eqnarray}
 m_{ab} + m_{ba}  \geq \Big[\mathcal E(g_S,g_S)\Big]^{-1} &=& \Big( \sum_{i,j} T_{ij} \Big)  \Big[{1\over 2} \sum_{i,j}T_{ij} \big(g_S(i) - g_S(j)\big)^2\Big]^{-1} \nonumber \\
&=& n  \Big[  \sum_{i \in S} \sum_{j \in S^c} T_{ij} \big(g_S(i) - g_S(j)\big)^2\Big]^{-1}
\nonumber \\
&=& \frac{n}{  \sum_{i \in S} \sum_{j \in S^c} T_{ij}}, \quad \textrm{for all } S \subset \mathcal
N, a \in S, b \notin S. \nonumber
\end{eqnarray}
The tightest lower bound can be obtained by taking the largest
right-hand side in the above relation, which gives the desired lower
bound.

For the upper bound, similar to Proposition 2 in Chapter 4 of
\cite{aldous}, we use the second characterization of the mean
commute time presented in Lemma \ref{extrmal_commute}. Note that any
unit flow from $a$ to $b$ is feasible in the minimization problem in
Eq.\  (\ref{Thompson}). Max-flow min-cut theorem implies that there
exists a flow $f$ of size $c_{ab}$ from $a$ to $b$ such that
$|f^*_{ij}| \leq T_{ij}$ for all edges $\{i,j\} \in \mathcal A$.
Therefore, there exists a unit flow $f = (f^*/c_{ab})$ from $a$ to
$b$ such that $|f_{ij}| \leq {T_{ij}}/{c_{ab}}$ for all edges
$\{i,j\}$. By deleting flows around cycles we may assume that
\begin{equation}\label{abs_flow_bd}
\sum_{l=1}^n|f_{kl}| \leq  \left\{
\begin{array}{cc}
1, & \text{if } k=a,b, \\
2,& \text{otherwise.}
\end{array}%
\right.
\end{equation}
Therefore, by invoking Lemma \ref{extrmal_commute} from Appendix D,
we obtain
\begin{eqnarray}\label{mincut_bd}
  m_{ab} + m_{ba} \leq \Big( \sum_{i,j} T_{ij} \Big) \sum_{\{i,j\} \in \mathcal A} \frac{f_{ij}^2}{T_{ij}} &\leq&  \frac{n}{c_{ab}} \sum_{\{i,j\} \in \mathcal A} |f_{ij}|\nonumber \\
  &\leq&\frac{n^2}{c_{ab}} , \nonumber
\end{eqnarray}
where the last inequality follows from (\ref{abs_flow_bd}).
\end{proof}

The minimum relative cut for the barbell graph in Example \ref{barbell_ex}
is given by cut $C_1$ with relative cut value $O(1)$. An alternative
relative cut between the forceful and influenced agents that partitions the
left bell is cut $C_3$, which has relative cut value $O(n)$, and therefore
yields a tighter bound on the mean commute times. Comparing cut $C_1$ to cut
$C_3$, we observe that $C_3$ is a \emph{balanced cut}, i.e., it partitions
the graph into parts each with a fraction of the total number of nodes,
while cut $C_1$ is not balanced. In order to avoid unbalanced cuts, we
introduce the notion of a \emph{normalized relative cut} between two nodes
which is a generalization of the normalized cut presented in Definition \ref%
{cond_def}.

\begin{definition}
\label{rel_cond_def} Consider a Markov chain with set of states $\mathcal{N}$%
, transition probability matrix $Z$, and stationary distribution $\pi$. The
\emph{minimum normalized relative cut value between $a$ and $b$}, denoted by
$\rho_{ab}$, is defined as
\begin{equation}
\rho_{ab} = \inf_{S \subset \mathcal{N }} \Big\{ \frac{Q(S,S^c)}{\pi(S) \pi(S^c)}\ |\ a \in S, b
\notin S\Big\},\nonumber
\end{equation}
where $Q(A,B) = \sum_{i \in A, j \in B} \pi_i Z_{ij}$, and $\pi(S) = \sum_{i
\in S} \pi_i$. We refer to the cut that achieves the minimum in this
optimization problem as the \textit{minimum normalized relative cut}.
\end{definition}

The next theorem provides a bound on the mean commute time between two nodes
$a$ and $b$ as a function of the minimum normalized relative cut value
between $a$ and $b$.

\begin{theorem}
\label{rel_cut_thm} Consider a Markov chain with set of states $\mathcal{N}$%
, transition probability matrix $Z$, and uniform stationary distribution.
For any $a, b \in \mathcal{N}$, we have
\begin{equation*}  \label{commute_cond}
m_{ab}+m_{ba} \leq \frac{3n\log n}{\rho_{ab}},
\end{equation*}
where $\rho_{ab}$ is the minimum normalized relative cut value between $a$
and $b$ (cf.\ Definition \ref{rel_cond_def}).
\end{theorem}

\begin{proof}
We present a generalization of the proof of Lemma \ref{cond_lemma}
by Aldous and Fill \cite{aldous}, for the notion of normalized
relative cuts. The proof relies on the characterization of the mean
commute time given by Lemma \ref{extrmal_commute} in Appendix D. For
a function $0 \leq g\leq 1$ with $g(a) = 0$ and $g(b) = 1$, order
the nodes as $a=1,2,\ldots,n=b$ so that $g$ is increasing. The
Dirichlet form (cf.\ Definition \ref{dirichlet_form}) can be written
as
\begin{eqnarray}\label{extr_1}
  \mathcal E(g,g) &=& \sum_{i}\sum_{k>i} \pi_i Z_{ik} \big(g(k) - g(i)\big)^2 \nonumber \\
  &\geq& \sum_{i}\sum_{k>i}\sum_{i\leq j < k}\pi_i Z_{ik} \big(g(j+1) - g(j)\big)^2 \nonumber \\
  &=& \sum_{j=1}^{n-1}  \big(g(j+1) - g(j)\big)^2 Q(A_j, A_j^c) \nonumber \\
  &\geq&   \sum_{j=1}^{n-1}  \big(g(j+1) - g(j)\big)^2 \rho_{ab} \pi(A_j) \pi(A_j^c),
\end{eqnarray}
where $A_j = \{1,2, \ldots, j\}$, and the last inequality is true by Definition \ref{rel_cond_def}.
On the other hand, we have
$$1 = g(b)-g(a) = \sum_{j=1}^{n-1}  \big(g(j+1) - g(j)\big) \big(\rho_{ab} \pi(A_j) \pi(A_j^c)\big)^{\frac{1}{2}}  \big(\rho_{ab} \pi(A_j) \pi(A_j^c)\big)^{-\frac{1}{2}}.$$
Using the Cauchy-Schwartz inequality and Eq.\ (\ref{extr_1}), we
obtain
\begin{equation}\label{extr_2}
    \frac{1}{\mathcal E(g,g)}\leq \frac{1}{\rho_{ab}}  \sum_{j=1}^{n-1} \frac{1}{\pi(A_j)\pi(A_j^c)}.
\end{equation}
But $\pi(A_j) = j/n$, because the stationary distribution of the Markov chain is uniform. Thus,
$$\sum_{j=1}^{n-1} \frac{1}{\pi(A_j)\pi(A_j^c)} = \sum_{j=1}^{n-1} \frac{n^2}{j(n-j)} \leq 3n\log n.$$
Therefore, by applying the above relation to Eq.\ (\ref{extr_2}) we
conclude
$$ \frac{1}{\mathcal E(g,g)}\leq \frac{3n\log n}{\rho_{ab}}.$$

The above relation is valid for every function $g$ feasible for the
maximization problem in Eq.\ (\ref{Dirichlet}). Hence, the desired
result follows from the extremal characterization of the mean
commute time given by Lemma \ref{extrmal_commute}.
\end{proof}

Note that the minimum normalized cut value of a Markov chain in Definition %
\ref{cond_def} can be related to normalized relative cut values as follows:
\begin{equation*}
\rho = \inf_{a \neq b \in \mathcal{N}} \{ \rho_{ab}\}.
\end{equation*}
Therefore, the upper bound given in Theorem \ref{rel_cut_thm} for the mean
commute time is always tighter than that provided in Lemma \ref{cond_lemma}.

Let us now examine our new characterization in the context of Example \ref%
{barbell_ex}. The minimum normalized relative cut is given by cut
$C_2$ with (normalized relative cut) value $O(1)$. Despite the fact
that $C_2$ is a balanced cut with respect to the entire graph, it is
not a balanced cut in the left bell. Therefore, it yields a worse
upper bound on mean commute times compared to cut $C_3$ [which has
value $O(n)$]. These considerations motivate us to consider balanced
cuts within subsets of the original graph. In the following we
obtain tighter bounds on the mean commute times by considering
relative cuts in a subset of the original graph.

\begin{definition}
\label{subgraph_def} Consider a weighted undirected graph, $(\mathcal{N},
\mathcal{A})$, with edge $\{i,j\}$ weight given by $w_{ij}$. For any $S
\subseteq \mathcal{N}$, we define the \emph{subgraph of $(\mathcal{N},
\mathcal{A})$ with respect to $S$} as a weighted undirected graph, denoted
by $(S, \mathcal{A}_S)$, where $\mathcal{A}_S$ contains all edges of the
original graph connecting nodes in $S$ with the following weights
\begin{equation}  \label{sub_weights}
\bar w_{ij} = \left\{
\begin{array}{ll}
w_{ij}, & \hbox{$i \neq j$;} \\
w_{ii}+\sum_{k \in S^c} w_{ik}, & \hbox{$i = j$.}%
\end{array}
\right. \nonumber
\end{equation}
\end{definition}

The next lemma uses the Monotonicity Law presented in Appendix D,
Lemma \ref{monotonicity_lemma} to relate the mean commute times
within a subgraph to the mean commute times of the original graph.

\begin{lemma}
\label{sub_commute_lemma} \emph{Let $G=(N,A)$ be an undirected graph
with edge $\{i,j\}$ weight given by $w_{ij}$. Consider a Markov
chain induced by this graph and denote the mean first passage times
between states $i$ and $j$ by $m_{ij}$. We fix nodes $a,b\in N$, and
$S\subseteq N$ containing $a$ and $b$. Consider a subgraph of
$(N,A)$ with respect to $S$ (cf. Definition \ref{subgraph_def}) and
let $\bar{m}_{ij}$ denote the mean first passage time between states
$i$ and $j$ for the Markov chain induced by this subgraph. We have,
\begin{equation}
m_{ab}+m_{ba}\leq \frac{w}{w(S)}(\bar{m}_{ab}+\bar{m}_{ba}), \nonumber\label{sub_commute_bd}
\end{equation}%
where $w$ is the total edge weight of the original graph, and $w(S)$ is the
total edge weight of the subgraph, i.e., $w(S)=\sum_{i\in S}\sum_{j\in
\mathcal{N}}w_{ij}$.}
\end{lemma}

\begin{proof}
Consider an undirected graph $(\mathcal N, \mathcal A)$ with
modified edge weights $\tilde w_{ij}$ given by
$$\tilde w_{ij} = \left\{
                    \begin{array}{ll}
                      w_{ij}, & \hbox{$i \neq j \in S,$ or $i \neq j \in S^c$ ;} \\
                      0, & \hbox{$i \in S, j \in S^c$;} \\
                      w_{ii} + \sum_{k \in S^c} w_{ik}, & \hbox{$i=j$.}
                    \end{array}
                  \right.
$$
Hence, $\tilde w_{ij} \leq w_{ij}$ for all $i \neq j$, but the total
edge weight $w$ remains unchanged. By Monotonicity Law (cf.\ Lemma
\ref{monotonicity_lemma}), the mean commute time in the original
graph is bounded by that of the modified graph, i.e.,
\begin{equation}\label{mono_commute_bd}
    m_{ab}+m_{ba} \leq \tilde m_{ab} + \tilde m_{ba}.
\end{equation}

The mean commute time in the modified graph can be characterized
using Lemma \ref{extrmal_commute} in terms of the Dirichlet form
defined in Definition \ref{dirichlet_form}. In particular,
\begin{eqnarray}\label{mono_commute_RHS}
  (\tilde m_{ab} +\tilde m_{ba})^{-1} &=& \inf_{0\leq g\leq 1} \Big\{\frac{1}{w} \sum_{i,j \in \mathcal N} \tilde w_{ij} \big(\tilde g(i) - \tilde g(j)\big)^2:\ g(a) = 0, g(b) = 1 \Big\}\nonumber \\
  &=& \inf_{0\leq g\leq 1} \Big\{\frac{1}{w} \sum_{i,j \in S}  w_{ij} \big(\tilde g(i) - \tilde g(j)\big)^2:\ g(a) = 0, g(b) = 1 \Big\}\nonumber \\
  && +\inf_{0\leq g\leq 1} \Big\{\frac{1}{w} \sum_{i,j \in S^c}  w_{ij} \big(\tilde g(i) - \tilde g(j)\big)^2 \Big\}\nonumber \\
  &=& \frac{w(S)}{w} \inf_{0\leq g\leq 1} \Big\{ \sum_{i,j \in S}  \frac{\bar w_{ij}}{w(S)} \big(\tilde g(i) - \tilde g(j)\big)^2:\ g(a) = 0, g(b) = 1
  \Big\} \nonumber \\
  &=& \frac{w(S)}{w} (\bar m_{ab} +\bar m_{ba})^{-1}, \nonumber
\end{eqnarray}
where the second equality holds by definition of $\tilde w$, and the last equality is given by
definition of $\bar w$, and the extremal characterization of the mean commute time in the subgraph.
The desired result is established by combining the above relation with (\ref{mono_commute_bd}).
\end{proof}

\begin{theorem}
\label{sub_cond_thm} Let $\mathcal{G}=(\mathcal{N}, \mathcal{A})$ be the
social network graph induced by the social network matrix $T$ and
consider a Markov chain with transition matrix $T$. For any $a, b \in
\mathcal{N}$, and any $S \subseteq \mathcal{N}$ containing $a$ and $b$, we
have
\begin{equation*}  \label{commute_cond}
m_{ab}+m_{ba} \leq \frac{3n\log |S|}{\rho_{ab}(S)},
\end{equation*}
where $\rho_{ab}(S)$ is the \textit{minimum normalized cut value between $a$
and $b$ on the subgraph of $(\mathcal{N}, \mathcal{A})$ with respect to $S$}%
, i.e.,
\begin{equation}  \label{sub_cond_def}
\rho_{ab}(S) = \inf_{S^{\prime }\subset S} |S|\frac{\sum_{i \in S^{\prime },
j \in S \setminus S^{\prime }} T_{ij}}{|S^{\prime }|\cdot|S\setminus
S^{\prime }|}.
\end{equation}
\end{theorem}

\begin{proof}
By Lemma \ref{sub_commute_lemma}, we have
\begin{eqnarray}\label{sub_comm_uniform}
    m_{ab}+m_{ba} &\leq&  \frac{w}{w(S)} (\bar m_{ab}+\bar m_{ba} ) \nonumber \\
    &=& \frac{n}{|S|} (\bar m_{ab}+\bar m_{ba} ),
\end{eqnarray}
where $\bar m_{ab}$ is the mean first passage time on the subgraph $(S, \mathcal A_S)$.

On the other hand, Definition \ref{subgraph_def} implies that for the subgraph $(S, \mathcal A_S)$,
we have for every $i \in S$
$$\sum_{k \in S} \bar w_{ik} = \sum_{k \in S\setminus\{i\} } \bar w_{ik} + \bar w_{ii} = \sum_{k \in \mathcal N} w_{ik}=\sum_{k \in \mathcal N} T_{ik} = 1.$$
Hence, the stationary distribution of the Markov chain on the
subgraph is uniform. Therefore, we can apply Lemma \ref{rel_cut_thm}
to relate the mean commute time within the subgraph $(S, \mathcal
A_S)$ to its normalized relative cuts, i.e.,
$$\bar m_{ab}+\bar m_{ba} \leq \frac{3|S|\log |S|}{\rho_{ab}(S)},$$
where $\rho_{ab}(S)$ is the minimum normalized cut between $a$ and
$b$ given by Definition \ref{rel_cond_def} on the subgraph. Since
the stationary distribution of the random walk on the subgraph is
uniform, we can rewrite $\rho_{ab}(S)$ as in (\ref{sub_cond_def}).
Combining the above inequality with Eq.\ (\ref{sub_comm_uniform})
establishes the theorem.
\end{proof}

Theorem \ref{sub_cond_thm} states that if the local neighborhood
around the forceful links are highly connected, the mean commute
times between the forceful and the influenced agents will be
small, implying a smaller excess influence for all agents, hence
limited spread of misinformation in the society. The economic
intuition for this result is similar to that for our main
characterization theorems: forceful agents get (their limited)
information intake from their local neighborhoods. When these
local neighborhoods are also connected to the rest of the network,
forceful agents will be indirectly influenced by the rest of the
society and this will limit the spread of their (potentially
extreme) opinions. In contrast, when their local neighborhoods
obtain most of their information from the forceful agents, the
opinions of these forceful agents will be reinforced (rather than
moderated) and this can significantly increase their excess
influence and the potential spread of misinformation.

Let us revisit Example \ref{barbell_ex}, and apply the result of Theorem \ref%
{sub_cond_thm} where the selected subgraph is the left cluster of nodes. The
left bell is approximately a complete graph. We observe that the minimum
normalized cut in the subgraph would be of the form of $C_3$ in Figure \ref%
{barbell_fig}, and hence the upper bound on the mean commute time between $i$
and $j$ is $O(n\log n)$, which is close to the mean commute time on a
complete graph of size $n$.

Note that it is possible to obtain the tightest upper bound on mean
commute time between two nodes by minimizing the bound in Theorem
\ref{sub_cond_thm} over all subgraphs $S$ of the social network
graph. However, exhaustive search over all subgraphs is not
appealing from a computational point of view. Intuitively, for any
two particular nodes, the goal is to identify whether such nodes are
highly connected by identifying a cluster of nodes containing them,
or a bottleneck that separates them. In the following section we
present a hierarchical clustering method to obtain such a cluster
using a recursive approach.

\subsubsection{Graph Clustering}

We next present a graph clustering method to provide tighter
bounds on the mean commute time between two nodes $a$ and $b$ by
systematically searching over subgraphs $S$ of the social network
graph that would yield improved normalized cut values. The goal of
this exercise is again to improve the bounds on the term
$(m_{ik}-m_{jk})$ in Eq. (\ref{exact_sensitivity_disjoint}) in
Theorem \ref{sensitivity_thm}.

The following algorithm is based on successive graph cutting using
the notion of minimum normalized cut value defined in Definition
\ref{cond_def}. This approach is similar to the graph partitioning
approach of Shi and Malik \cite{shimalik} applied to image
segmentation problems.

\begin{algorithm}
\label{cluster_alg} Fix nodes $a, b$ on the social network
graph $(\mathcal{N}, \mathcal{A})$. Perform the following steps:

\begin{enumerate}
\item $k = 0$, $S_k = \mathcal{N}$.

\item Define $\rho_k$ as
\begin{equation*}
\rho_k = \inf_{S \subset S_k} |S_k|\frac{\sum_{i \in S, j \in S_k \setminus
S} T_{ij}}{|S|\cdot|S_k\setminus S|},
\end{equation*}
with $S^*_k$ as an optimal solution.

\item If $a,b \in S^*_k$, then $S_{k+1} = S^*_k$; $k \leftarrow k+1$; Goto 2.

\item If $a,b \in S_k\setminus S^*_k$, then $S_{k+1} = S_k \setminus S^*_k$;
$k \leftarrow k+1$; Goto 2.

\item Return $\frac{3 n\log|S_k|}{\rho_k}$.
\end{enumerate}
\end{algorithm}

\begin{figure}[tbp]
\centering
\includegraphics[width=.6\textwidth]{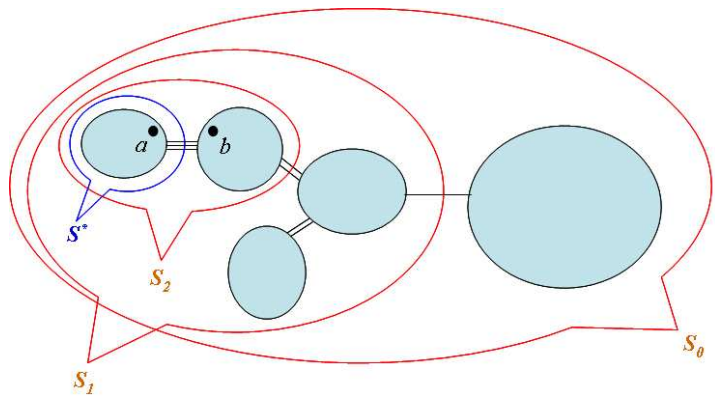}\newline
\caption{Graph clustering algorithm via successive graph cutting
using normalized minimum cut criterion.} \label{alg_fig}
\end{figure}

Figure \ref{alg_fig} illustrates the steps of Algorithm
\ref{cluster_alg} for a highly clustered graph. Each of the regions
in Figure \ref{alg_fig} demonstrate a highly connected subgraph. We
observe that the global cut given by $S_1$ does not separate $a$ and
$b$, so it need not give a tight characterization of the bottleneck
between $a$ and $b$. Nevertheless, $S_1$ gives a better estimate of
the cluster containing $a$ and $b$. Repeating the above steps, the
cluster size reduces until we obtain a normalized cut separating $a$
and $b$. By Theorem \ref{sub_cond_thm}, this cut provides a bound on
the mean commute time between $a$ and $b$ that characterizes the
bottleneck between such nodes. So far, we have seen in this example
and Example \ref{barbell_fig} that graph clustering via recursive
partitioning can monotonically improve upon the bounds on the excess
influence (cf.\ Theorem \ref{sub_cond_thm}). Unfortunately, that is
not always the case as discussed in the following example. In fact,
we need further assumptions on the graph in order to obtain monotone
improvement via graph clustering.

\begin{example}
\begin{figure}[tbp]
\centering
\includegraphics[width=.6\textwidth]{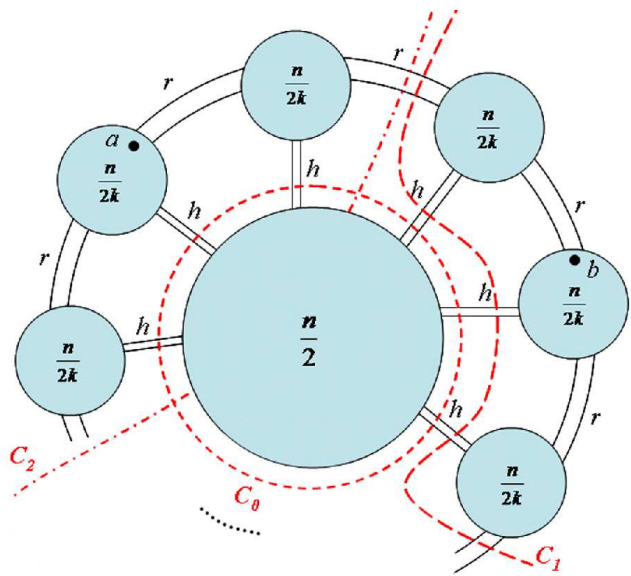}\newline
\caption{Social network graph with a central hub}
\label{hub_fig}
\end{figure}

Consider a social network graph of size $n$ depicted in
Figure \ref{hub_fig}. The central region is a complete graph of size $n/2$.
Each of the $k$ clusters on the cycle is a complete graph of size $n/(2k)$,
which is connected to the central hub via edges of total weight $h$.
Moreover, the clusters on the cycle are connected with total edge weight $r$.

If $r \geq kh/8$, then $C_0$ would be the minimum normalized cut rather than
cuts of the form $C_1$. Hence, $\rho_0$ in step 2 of Algorithm \ref%
{cluster_alg} is given by
\begin{equation*}
\rho_0 = n \frac{kh}{\frac{n}{2}\cdot \frac{n}{2}} = \frac{4kh}{n}.
\end{equation*}

After removing the central cluster, we obtain $C_2$ as the minimum
normalized cut over the cycle, with the following value
\begin{equation*}
\rho_1 = {\frac{n}{2}} \frac{2r}{\frac{n}{4}\cdot \frac{n}{4}} = \frac{16r}{n%
}.
\end{equation*}

Therefore, we conclude that $\rho_1 < \rho_0$ if and only if $\frac{kh}{8} <
r < \frac{kh}{4}$, i.e., the upperbound obtained by Algorithm \ref%
{cluster_alg} on the mean commute time between $a$ and $b$, is not smaller
than that of Lemma \ref{cond_lemma}. That is because by removing the central
cluster, we have eliminated the possibility of reaching the destination via
shortcuts of the central hub, and the only way to reach the destination is
to walk through the cycle.
\end{example}

Next, we show that the bounds given by Algorithm \ref{cluster_alg} are
monotonically improving, if the successive cuts are disjoint.

\begin{definition}
\label{disjoint_cut_def} Consider an undirected graph $(\mathcal{N},
\mathcal{A})$. The cuts defined by $S_1, S_2 \subseteq \mathcal{N}$ are
\emph{disjoint with respect to $\mathcal{N}$} if
\begin{equation*}
\delta(S_1) \cap \delta(S_2) = \emptyset,
\end{equation*}
where
\begin{equation*}
\delta(S) = \Big\{\{i,j\}\in \mathcal{A}\ |\ i \in S, j \in S^c \Big\}.
\end{equation*}
\end{definition}

\begin{theorem}
\label{monotone_bd_thm} Let $\rho_k$ and $S_k$ be generated by the $k^{th}$
iteration of running Algorithm \ref{cluster_alg} on the social network graph $(\mathcal{N}, \mathcal{A})$. If the cuts corresponding
to $S_{k+1}$ and $S_{k+2}$ are disjoint with respect to $S_k$, then $%
\rho_{k+1} > \rho_{k}$.
\end{theorem}

\begin{proof}
By definition of $\rho_{k}$ in step 2 of Algorithm \ref{cluster_alg}, we have for $S_{k+2}
\subseteq S_k$
\begin{equation}\label{mono_pf1}
    \rho_k = |S_k|\frac{\sum_{i \in S_{k+1}, j \in S_k \setminus S_{k+1}} T_{ij}}{|S_{k+1}|\cdot|S_k \setminus
    S_{k+1}|} \leq |S_k|\frac{\sum_{i \in S_{k+2}, j \in S_k \setminus S_{k+2}} T_{ij}}{|S_{k+2}|\cdot|S_k \setminus
    S_{k+2}|}.
\end{equation}
But $S_{k+1}$ and $S_{k+2}$ are disjoint with respect to $S_k$, and
$S_{k+2} \subseteq S_{k+1} \subseteq S_k$. It is straightforward to
show that
$$\Big\{\{i,j\}\in \mathcal A\ |\ i \in S_{k+2}, j \in S_k\setminus S_{k+1}\Big\} \subseteq \delta(S_{k+1}) \cap \delta(S_{k+2})= \emptyset,$$
which implies
$$\sum_{i \in S_{k+2}, j \in S_k \setminus S_{k+2}} T_{ij} = \sum_{i \in S_{k+2}, j \in S_k \setminus S_{k+1}} T_{ij} + \sum_{i \in S_{k+2}, j \in S_{k+1} \setminus S_{k+2}} T_{ij} = \sum_{i \in S_{k+2}, j \in S_{k+1} \setminus S_{k+2}} T_{ij}.$$
Therefore, by combining the above relation with (\ref{mono_pf1}) and
the  definition of  $\rho_{k+1}$, we obtain
\begin{eqnarray}
  \frac{\rho_{k+1}}{\rho_{k}} &\geq&  \bigg(|S_{k+1}|\frac{\sum_{i \in S_{k+2}, j \in S_{k+1} \setminus S_{k+2}} T_{ij}}{|S_{k+2}|\cdot|S_{k+1} \setminus
  S_{k+2}|}\bigg) \bigg(|S_{k}|\frac{\sum_{i \in S_{k+2}, j \in S_{k+1} \setminus S_{k+2}} T_{ij}}{|S_{k+2}|\cdot|S_{k} \setminus
  S_{k+2}|}\bigg)^{-1} \nonumber \\
  &=& \frac{|S_{k+1}|\cdot|S_k\setminus S_{k+2}|}{|S_{k}|\cdot|S_{k+1} \setminus S_{k+2}|} =  \frac{|S_{k+1}|\big(|S_k\setminus
  S_{k+1}|+|S_{k+1}\setminus S_{k+2}|\big)}{|S_{k+1} \setminus S_{k+2}|\big(|S_{k+1}|+|S_{k}\setminus
  S_{k+1}|\big)} \label{mono_pf2}\\
  &=& \bigg(1+\frac{|S_k\setminus  S_{k+1}|}{|S_{k+1}\setminus S_{k+2}|}\bigg)\bigg(1+\frac{|S_k\setminus
  S_{k+1}|}{|S_{k+1}|}\bigg)^{-1} \nonumber \\
  &>& 1, \nonumber
\end{eqnarray}
where (\ref{mono_pf2}) holds because $S_{k+2} \subseteq S_{k+1} \subseteq S_k$, and the last
inequality is true because $S_{k+1}\setminus S_{k+2} \subseteq S_{k+1}$, and $S_{k+2}$ is nonempty.
\end{proof}

\section{Conclusions}

\label{conclusions}

This paper analyzed the spread of misinformation in large
societies. Our analysis is motivated by the widespread differences
in beliefs across societies and more explicitly, the presence of
many societies in which beliefs that appear to contradict the
truth can be widely held. We argued that the possibility that such
misinformation can arise and spread is the manifestation of the
natural tension between information aggregation and misinformation
spreading in the society.

We modeled a society as a social network of agents communicating
(meeting) with each other. Each individual holds a belief
represented by a scalar. Individuals meet pairwise and exchange
information, which is modeled as both individuals adopting the
average of their pre-meeting beliefs. When all individuals engage
in this type of information exchange, the society will be able to
aggregate the initial information held by all individuals. This
effective information aggregation forms the benchmark against
which we compared the possible spread of misinformation.

Misinformation is introduced by allowing some agents to be \textquotedblleft
forceful,\textquotedblright\ meaning that they influence the beliefs of
(some) of the other individuals they meet, but do not change their own
opinion. When the influence of forceful agents is taken into account, this
defines a stochastic process for belief evolution, and our analysis
exploited the fact that this stochastic process (Markov chain) can be
decomposed into a part induced by the social network matrix and a part
corresponding to the influence matrix.

Under the assumption that even forceful agents obtain some information
(however infrequent) from some others, we first show that beliefs in this
class of societies converge to a consensus among all individuals (under some
additional weak regularity conditions). This consensus value is a random
variable, and the bulk of our analysis characterizes its behavior, in
particular, providing bounds on how much this consensus can differ from the
efficient information aggregation benchmark.

We presented three sets of results. Our first set of results quantify the extent of misinformation
in the society as a function of the number and properties of forceful agents and the mixing
properties of the Markov chain induced by the social network matrix. In particular, we showed that
social network matrices with large second eigenvalues, or that correspond to fast-mixing graphs,
will place tight bounds on the extent of misinformation. The intuition for this result is that in
such societies individuals that ultimately have some influence on the beliefs of forceful agents
rapidly inherit the beliefs of the rest of the society and thus the beliefs of forceful agents
ultimately approach to those of the rest of the society and cannot have a large impact on the
consensus beliefs. The extreme example is provided by expander graphs, where, when the number and
the impact of forceful agents is finite, the extent of misinformation becomes arbitrarily small as
the size of the society becomes large. In contrast, the worst outcomes are obtained when there are
several forceful agents and forceful agents themselves update their beliefs only on the basis of
information they obtain from individuals most likely to have received their own information
previously (i.e., when the graph is slow-mixing).

Our second set of results exploit more explicitly the location of
forceful agents within a social network. A given social network will
lead to very different types of limiting behavior depending on the
context in which the forceful agents are located. We provided a
tight characterization for graphs with the forceful essential edges,
that is, graphs representing societies in which a forceful agent
links two disconnected clusters. Such graphs approximate situations
in which forceful agents, such as media outlets or political
leaders, themselves obtain all of their information from a small
group of other individuals. The interesting and striking result in
this case is that the excess influence of all of the members of the
small group are the same, even if some of them are not directly
linked to forceful agents. We then extended these findings to more
general societies using the notion of information bottlenecks.

Our third set of results provide new efficient graph clustering
algorithms for computing tighter bounds on excess influence.

We view our paper as a first attempt in quantifying misinformation
in society. As such, we made several simplifying assumptions and
emphasized the characterization results to apply for general
societies. Many areas of future investigation stem from this
endeavor. First, it is important to consider scenarios in which
learning and information updating are, at least partly, Bayesian.
Our non-Bayesian framework is a natural starting point, both
because it is simpler to analyze and because the notion of
misinformation is more difficult to introduce in Bayesian models.
Nevertheless, game theoretic models of communication can be used
for analyzing situations in which a sender may explicitly try to
mislead one or several receivers. Second, one can combine a model
of communication along the lines of our setup with individuals
taking actions with immediate payoff consequences and also
updating on the basis of their payoffs. Misinformation will then
have short-run payoff consequences, but whether it will persist or
not will depend on how informative payoffs are and on the severity
of its short-run payoff consequences. Third, it would be useful to
characterize what types of social networks are more robust to the
introduction of misinformation and how agents might use simple
rules in order to avoid misinformation.

Finally, our approach implies that the society (social network)
will ultimately reach a consensus, even though this consensus
opinion is a random variable. In practice, there are widespread
differences in beliefs in almost all societies. There is little
systematic analysis of such differences in beliefs in the
literature at the moment, and this is clearly an important and
challenging area for future research. Our framework suggests two
fruitful lines of research. First, although a stochastic consensus
is eventually reached in our model, convergence can be very slow.
Thus characterizing the rate of convergence to consensus in this
class of models might provide insights about what types of
societies and which sets of issues should lead to such belief
differences. Second, if we relax the assumption that even forceful
agents necessarily obtain some (albeit limited) information from
others, thus removing the ``no man is an island'' feature, then it
can be shown that the society will generally not reach a
consensus. Nevertheless, characterizing differences in opinions in
this case is difficult and requires a different mathematical
approach. We plan to investigate this issue in future work.

\newpage

\noindent \textbf{\large Appendix A}

\vskip .5pc

\noindent\textbf{{\large Preliminary Lemmas, Sections \ref{convergence} and %
\ref{global-misinfo}}}

\vskip .5pc

This appendix presents two lemmas that will be used in proving the
convergence of agent beliefs (i.e., Theorem \ref{convconsensus}) and in
establishing properties of the social network matrix $T$ in Appendix C.

The first lemma provides conditions under which a nonnegative $n\times n$
matrix $M$ is \textit{primitive}, i.e., there exists a positive integer $k$
such that all entries of the $k^{th}$ power of $M$, $M^k$, are positive (see
\cite{seneta}). The lemma also provides a positive uniform lower bound on
the entries of the matrix $M^k$ as a function of the entries of $M$ and the
properties of the graph induced by the positive entries of matrix $M$. A
version of this lemma was established in \cite{distpaper}. We omit the proof
here since it is not directly relevant to the rest of the analysis.

\begin{lemma}
\emph{Let $H$ be a nonnegative $n\times n$ matrix that satisfies the
following conditions:
\begin{itemize}
\item[(a)] The diagonal entries of $H$ are positive, i.e., $H_{ii}>0$ for all $i$.
\item[(b)] Let ${\cal E}$ denote a set of edges such that the
graph $({\cal N}, {\cal E})$ is connected. For all $(i,j) \in {\cal
E}$, the entry $H_{ij}$ is positive, i.e., ${\cal E}\subset \{(i,j)\
|\ H_{ij}>0\} $.
\end{itemize}
Let $d$ denote the maximum shortest path length between any $i,j$ in
the induced graph $({\cal N},{\cal E})$, and $\eta>0$ be a scalar
given by
\[\eta=\min\left\{ \min_{i\in {\cal N}} H_{ii}, \min_{(i,j)\in {\cal E}}
H_{ij} \right\}.\] Then, we have
\[[H^d]_{ij}\ge \eta^d\qquad \hbox{for all }i,j.\]}
\label{primitive}
\end{lemma}

The second lemma considers a sequence $z(k)$ generated by a linear
time-varying update rule, i.e., given some $z(0)$, the sequence $\{z(k)\}$
is generated by
\begin{equation*}
z(k)=H(k) z(k-1)\qquad \hbox{for all }k\ge 0,
\end{equation*}
where $H(k)$ is a stochastic matrix for all $k\ge 0$. We introduce the
matrices $\tilde \Phi(k,s) = H(k)H(k-1)\ldots H(s)$ to relate $z(k+1)$ to $%
z(s)$ for $s\le k$, i.e.,
\begin{equation*}
z(k+1)=\tilde \Phi(k,s)z(s).
\end{equation*}
The lemma shows that, under some assumptions on the entries of the matrix $%
\tilde \Phi(k,s)$, the \textit{disagreement in the components} of $z(k)$,
defined as the difference between the maximum and minimum components of $%
z(k) $, decreases with $k$ and provides a bound on the amount of decrease.

\begin{lemma}
\emph{Let $\{H(k)\}$ be a sequence of $n\times n$ stochastic matrices. Given
any $z(0)\in \mathbb{R}^n$, let $\{z(k)\}$ be a sequence generated by the
linear update rule
\begin{equation}
z(k)=H(k) z(k-1)\qquad \hbox{for all }k\ge 0.  \label{zseq}
\end{equation}
Assume that there exists some integer $B>0$ and scalar $\theta>0$ such that
\begin{equation*}
[\tilde \Phi(s+B-1,s)]_{ij}\ge \theta\qquad \hbox{for all }i,j, \hbox{ and }
s\ge0.
\end{equation*}
For all $k\ge 0$, define $M(k)\in \mathbb{R}$ and $m(k)\in \mathbb{R}$ as
follows:
\begin{equation}
M(k) = \max_{i\in \mathcal{N}}\, z_i(k),\qquad m(k) = \min_{i\in \mathcal{N}%
}\, z_i(k).  \label{twom}
\end{equation}
Then, for all $s\ge 0$, we have $n\theta\le 1$ and
\begin{equation*}
M(s+B)-m(s+B)\le (1-n\theta)(M(s)-m(s)).
\end{equation*}
}\label{lyapdec}
\end{lemma}

\begin{proof} In view of the linear update rule (\ref{zseq}), we
have for all $i$,
\[z_i(s+B)=\sum_{j=1}^n [\tilde \Phi(s+B-1,s)]_{ij}\, z_j(s)\qquad \hbox{for all }s\ge 0.\]
We rewrite the preceding relation as
\begin{equation}z_i(s+B)=\sum_{j=1}^n \theta z_j(s) + \sum_{j=1}^n
[\hat \Phi(s+B-1,s)]_{ij}\, z_j(s),\label{zevol}\end{equation} where
$[\hat \Phi(s+B-1,s)]_{ij} = [\tilde \Phi(s+B-1,s)]_{ij}-\theta$ for
all $i,j$. Since by assumption $[\tilde \Phi(s+B-1,s)]_{ij}\ge
\theta$ for all $i,j$, we have
\[[\hat \Phi(s+B-1,s)]_{ij}\ge 0\qquad \hbox{for all } i,j. \]
Moreover, since the matrices $H(k)$ are stochastic, the product
matrix $\Phi(s+B-1,s)$ is also stochastic, and therefore we have
\[\sum_{j=1}^n [\hat \Phi(s+B-1,s)]_{ij} = 1-n \theta\qquad \hbox{for all }i.\]
From the preceding two relations, we obtain $1-n \theta\ge 0$ and
\[(1-n\theta) m(s)\le \sum_{j=1}^n [\hat \Phi(s+B-1,s)]_{ij}\, z_j(k) \le (1-n\theta) M(s),\]
where $m(s)$ and $M(s)$ are defined in Eq.\ (\ref{twom}). Combining
this relation with Eq.\ (\ref{zevol}), we obtain for all $i$
\[(1-n\theta) m(s)\le z_i(s+B) - \sum_{j=1}^n \theta z_j(s) \le (1-n\theta) M(s).\]
Since this relation holds for all $i$, we have
\[(1-n\theta) m(s)\le m(s+B) - \sum_{j=1}^n \theta z_j(s), \]
\[M(s+B) - \sum_{j=1}^n \theta z_j(s) \le (1-n\theta) M(s),\]
from which we obtain
\[M(s+B)-m(s+B)\le (1-n\theta)(M(s)-m(s))\qquad\hbox{for all }s\ge 0.\]
\end{proof}

\vskip 1pc

\noindent \textbf{\large Appendix B}

\vskip .5pc

\noindent\textbf{{\large Properties of the Mean Interaction and
Transition Matrices, Sections \ref{convergence} and \ref{global-misinfo}}}

\vskip .5pc

We establish some properties of the mean interaction matrix $\tilde W$ and
the transition matrices $\Phi(k,s)$ under the assumptions discussed in
Section \ref{assumptions}. Recall that transition matrices are given by
\begin{equation}
\Phi(k,s) = W(k)W(k-1)\cdots W(s+1)W(s)\qquad
\hbox{for all $k$
and $s$ with }\ k\ge s,  \label{tranmat}
\end{equation}
with $\Phi(k,k)=W(k)$ for all $k$. Also note that the mean interaction
matrix is given by $\tilde W=E[W(k)]$ for all $k$. In view of the belief
update model (\ref{beliefupdate})-(\ref{updatematrix}), the entries of the
matrix $\tilde W$ can be written as follows. For all $i\in \mathcal{N}$, the
diagonal entries are given by
\begin{equation}
[\tilde W]_{ii} = 1-\frac{\sum_{j\ne i} (p_{ij}+p_{ji})}{n} + {\frac{1}{n}}%
\, \left[ \sum_{j\ne i} p_{ij}\, \Big({\frac{\beta_{ij}}{2}} + \alpha_{ij}
\epsilon + \gamma_{ij}\Big) +\sum_{j\ne i} p_{ji}\, \Big({\frac{\beta_{ji}}{2%
}} + \alpha_{ji} + \gamma_{ji}\Big)\right],  \label{diagelem}
\end{equation}
and for all $i\ne j\in \mathcal{N}$, the off-diagonal entries are given by
\begin{equation}
[\tilde W]_{ij} = {\frac{1}{n}}\, \left[p_{ij}\, \Big({\frac{\beta_{ij}}{2}}
+\alpha_{ij} (1-\epsilon)\Big) + p_{ji}\, {\frac{\beta_{ji}}{2}}\right].
\label{nondiagelem}
\end{equation}

Using the assumptions of Section \ref{assumptions}, Lemma \ref{primitive},
and the explicit expressions for the entries of the matrix $\tilde W$, we
have the following result.

\begin{lemma}
\emph{Let $d$ be the maximum shortest path length between any $i,j$ in the
graph $(\mathcal{N},\mathcal{E})$ [cf.\ Eq.\ (\ref{maxsp})], and $\eta$ be a scalar
given by
\begin{equation}
\eta = \min\left\{\min_{i\in \mathcal{N}}\, [\tilde W]_{ii}, \min_{(i,j)\in
\mathcal{E}}\, [\tilde W]_{ij}\right\},  \label{scalareta}
\end{equation}
[cf.\ Eqs.\ (\ref{diagelem}) and (\ref{nondiagelem})].
\begin{itemize}
\item[(a)] The scalar $\eta$ is positive and we have
\[[\tilde W^d]_{ij}\ge \eta^d\qquad \hbox{for all }i,j.\]
\item[(b)] We have
\[P\Big\{[\Phi(s+d-1,s)]_{ij}\ge \frac{\eta^d}{2}\Big\}\ge \frac{\eta^d}{2}\qquad
\hbox{for all }s\ge 0,\ i,\hbox{ and }j.\]
\end{itemize}
 }\label{tranprimitive}
\end{lemma}

\begin{proof}
\noindent (a) \ We show that under Assumptions \ref{comprob} and \ref{intprob}, the mean
interaction matrix $\tilde W$ has positive diagonal entries and the set ${\cal E}$ [cf.\ Eq.\
(\ref{edges})] is a subset of the link set induced by the positive elements of $\tilde W$. Together
with the Connectivity assumption, part (a) then follows from Lemma \ref{primitive}.

By Assumption \ref{comprob}, we have for all $i$, $\sum_{j\ne i}
p_{ij}=1$ and $p_{ij}\ge 0$ for all $j$. This implies that
$\sum_{j\ne i} p_{ji}\le n-1$ and therefore
\[1-\frac{\sum_{j\ne i} (p_{ij}+p_{ji})}{n}\ge 0\qquad \hbox{for all }i.\]
Since $\sum_{j\ne i} p_{ij}=1$ for all $i$, there exists some $j$
such that $p_{ij}>0$, i.e., $(i,j)\in {\cal E}$. In view of the information exchange model,
we have $\beta_{ij}>0$ or $\alpha_{ij}>0$ or $\gamma_{ij}>0$, implying that
\[p_{ij}\, \Big({\beta_{ij}\over 2}+ \alpha_{ij} \epsilon + \gamma_{ij}\Big)>0.\]
Combining the preceding two relations with Eq.\ (\ref{diagelem}), we
obtain \begin{equation}[\tilde W]_{ii}>0\qquad \hbox{for all  }
i.\label{posdi}\end{equation}

We next show that for any link $(i,j)$ in the set ${\cal E}$, the entry $[\tilde W]_{ij}$ is
positive, i.e.,
\[{\cal E}\subset \{(i,j)\ |\ [\tilde W]_{ij}>0\}.\]
For any $(i,j)\in {\cal E}$, we have $p_{ij}>0$, and therefore
$\beta_{ij}+\alpha_{ij}>0$ (cf.\ Assumption \ref{intprob}). This
implies that
\[p_{ij}\, \Big({\beta_{ij}\over 2} +\alpha_{ij} (1-\epsilon)\Big)>0,\]
which by Eq.\ (\ref{nondiagelem}) yields $[\tilde W]_{ij}>0$.
Together with Eq.\ (\ref{posdi}), this shows that the scalar $\eta$
defined in (\ref{scalareta}) is positive. By Assumption
\ref{connect}, the graph $({\cal N,E})$ is connected. Using the
identification $H=\tilde W$ in Lemma \ref{primitive}, we see that
the conditions of this lemma are satisfied, establishing part (a).

\noindent (b) \ For all $i,j$ and $s\ge 0$, we have
\begin{eqnarray}
P\Big\{[\Phi(s+d-1,s)]_{ij}\ge \frac{\eta^d}{2}\Big\} &=&
P\Big\{1-[\Phi(s+d-1,s)]_{ij}\le 1- \frac{\eta^d}{2}\Big\}\nonumber \\
&=& 1- P\Big\{1-[\Phi(s+d-1,s)]_{ij}\ge 1-
\frac{\eta^d}{2}\Big\}.\label{Mapply}
\end{eqnarray}
The Markov Inequality states that for any nonnegative random
variable $Y$ with a finite mean $E[Y]$, the probability that the
outcome of the random variable $Y$ exceeds any given scalar
$\delta>0$ satisfies
\[{\rm P} \{Y\ge \delta\}\le \frac{E[Y]}{\delta}.\]
By applying the Markov inequality to the random variable
$1-[\Phi(s+d-1,s)]_{ij}$ [which is nonnegative and has a finite
expectation in view of the stochasticity of the matrix
$\Phi(s+d-1,s)$ for all $s\ge 0$], we obtain
\[P\Big\{1-[\Phi(s+d-1,s)]_{ij}\ge 1-
\frac{\eta^d}{2}\Big\} \le \frac{E[1-[\Phi(s+d-1,s)]_{ij}]}{1-
\eta^d/2}.\] Combining with Eq.\ (\ref{Mapply}), this yields
\begin{equation}P\Big\{[\Phi(s+d-1,s)]_{ij}\ge \frac{\eta^d}{2}\Big\}  \ge
1-\frac{E[1-[\Phi(s+d-1,s)]_{ij}]}{1-
\eta^d/2}.\label{yeter}\end{equation} By the definition of the
transition matrices [cf.\ Eq.\ (\ref{tranmat})], we have
\[E[\Phi(s+d-1,s)] = E[W(s+d-1)W(s+d-2)\cdots W(s)] = \tilde W^d,\]
where the second equality follows from the assumption that $W(k)$ is
independent and identically distributed over $k$. By part (a), this
implies that
\[[E[\Phi(s+d-1,s)]]_{ij}\ge \eta^d\qquad \hbox{for all }i,j,\]
which combined with Eq.\ (\ref{yeter}) yields
\begin{eqnarray*}
P\Big\{[\Phi(s+d-1,s)]_{ij}\ge \frac{\eta^d}{2}\Big\} \ge
1-\frac{1-\eta^d}{1-\eta^d/2} = \frac{\eta^d/2}{1-\eta^d/2} \ge
\frac{\eta^d}{2},
\end{eqnarray*}
establishing the desired result.
\end{proof}

The next two lemmas establish properties of transition matrices.

\begin{lemma}
\emph{
\begin{itemize}
\item[(a)] $[\Phi(k,s)]_{ii}\ge \epsilon^{k-s+1}$ for all $k$ and
$s$ with $k\ge s$, and all $i\in {\cal N}$ with probability one.
\item[(b)] Assume that there exist integers $K, B\ge 1$ and a scalar $\xi>0$ such that for
some $s\ge 0$ and $k\in \{0,\ldots,K\}$, we have
\[[\Phi(s+(k+1)B-1,s+kB)]_{ij}\ge \xi\qquad \hbox{for some }i,j.\]
Then,
\[[\Phi(s+KB-1,s)]_{ij}\ge \xi \epsilon^{K-1}\qquad \hbox{with probability one}.\]
\end{itemize}
}\label{tranpos}
\end{lemma}

\begin{proof}
\noindent{(a)} \ We let $s$ be arbitrary and prove the relation by
induction on $k$. By the definition of the transition matrices [cf.\
Eq.\ (\ref{tranmat})], we have $\Phi(s,s) = W(s)$. Thus, the
relation $[\Phi(k,s)]_{ii}\ge \epsilon^{k-s+1}$ holds for $k=s$ from
the definition of the update matrix $W(k)$ [cf.\ Eq.\
(\ref{updatematrix})]. Suppose now that the relation holds for some
$k>s$ and consider $[\Phi(k+1,s)]_{ii}$. We have
\[[\Phi(k+1,s)]_{ii} = \sum_{h=1}^n [W(k+1)]_{ih}[\Phi(k,s)]_{hi}\ge
[W(k+1)]_{ii}[\Phi(k,s)]_{ii}\ge \epsilon^{k-s+2},\] where the first
inequality follows from the nonnegativity of the entries of
$\Phi(k,s)$, and the second inequality follows from the inductive
hypothesis.

\noindent{(b)} \ For any $s\ge 0$, we have
\begin{eqnarray*}
[\Phi(s+KB-1,s)]_{ij}&=& \sum_{h=1}^n [\Phi(s+KB-1,s+(k+1)B)]_{ih}
[\Phi(s+(k+1)B-1,s)]_{hj}\\
&\ge& [\Phi(s+KB-1,s+(k+1)B)]_{ii} [\Phi(s+(k+1)B-1,s)]_{ij} \\&\ge&
\epsilon^{(K-k-1)B} [\Phi(s+(k+1)B-1,s)]_{ij},
\end{eqnarray*}
where the last inequality follows from part (a). Similarly,
\begin{eqnarray*}
[\Phi(s+(k+1)B-1,s)]_{ij} &=& \sum_{h=1}^n
[\Phi(s+(k+1)B-1,s+kB)]_{ih} [\Phi(s+kB-1,s)]_{hj} \\
&\ge& [\Phi(s+(k+1)B-1,s+kB)]_{ij} [\Phi(s+kB-1,s)]_{jj}\\& \ge& \xi
\epsilon^{kB},
\end{eqnarray*}
where the second inequality follows from the assumption
$[\Phi(s+(k+1)B-1,s+kB)]_{ij}\ge \xi$ and part (a). Combining the
preceding two relations yields the desired result.
\end{proof}

\begin{lemma}
\emph{We have
\begin{equation*}
P\left\{ [\Phi(s+ n^2d -1,s)]_{ij}\ge {\frac{\eta^d}{2}} \epsilon^{n^2-1},\ %
\hbox{for all }i,j\right\} \ge \left(\frac{\eta^d}{2}\right)^{n^2}\qquad %
\hbox{for all }s\ge 0,
\end{equation*}
where the scalar $\eta>0$ and the integer $d$ are the constants defined in
Lemma \ref{tranprimitive}.}\label{posprob}
\end{lemma}

\begin{proof} Consider a particular ordering of the elements of an $n\times
n$ matrix and let $k_{ij}\in \{0,\ldots,n^2-1\}$ denote the unique
index for element $(i,j)$. From Lemma \ref{tranpos}(b), we have
\begin{eqnarray*}
&&P\Big\{ [\Phi(s+ n^2d -1,s)]_{ij}\ge {\eta^d\over 2}
\epsilon^{n^2-1},\ \hbox{for all }i,j\Big\}
\\&&\qquad\qquad\qquad\ge P\Big\{ [\Phi(s+ (k_{ij}+1)d
-1,s+k_{ij}d)]_{ij}\ge {\eta^d\over 2},\
\hbox{for all }i,j\Big\} \\
&&\qquad\qquad\qquad=  \prod_{(i,j)} P\Big\{\Phi(s+ (k_{ij}+1)d
-1,s+k_{ij}d)]_{ij}\ge {\eta^d\over 2}\Big\}\\
&&\qquad\qquad\qquad\ge \left(\frac{\eta^d}{2}\right)^{n^2}.
\end{eqnarray*}
Here the second equality follows from the independence of the random
events
\[\Big\{\Phi(s+ (k+1)d -1,s+kd)]_{ij}\ge {\eta^d\over 2}\Big\}\] over
all $k=0,\ldots,n^2-1$, and the last inequality follows from Lemma
\ref{tranprimitive}(b).
\end{proof}

\vskip 1pc

\noindent \textbf{\large Appendix C}

\vskip .5pc

\noindent\textbf{{\large Properties of the Social Network Matrix, Section %
\ref{global-misinfo}}}

\vskip .5pc

The next lemma studies the properties of the social network matrix $T$. Note
that the entries of the matrix T can be written as follows: For all $i\in
\mathcal{N}$, the diagonal entries are given by
\begin{equation}
\lbrack T]_{ii}=1-\frac{\sum_{j\neq i}(p_{ij}+p_{ji})}{n}+{\frac{1}{n}}\,%
\left[ \sum_{j\neq i}p_{ij}\,\Big({\frac{1-\gamma _{ij}}{2}}+\gamma _{ij}%
\Big)+\sum_{j\neq i}p_{ji}\,\Big({\frac{1-\gamma _{ji}}{2}}+\gamma _{ji}\Big)%
\right] ,  \label{Tdiagelem}
\end{equation}%
and for all $i\neq j\in \mathcal{N}$, the off-diagonal entries are given by
\begin{equation}
\lbrack T]_{ij}={\frac{1}{n}}\,\left[ p_{ij}\,{\frac{1-\gamma _{ij}}{2}}%
+p_{ji}\,{\frac{1-\gamma _{ji}}{2}}\right] .  \label{Toffelem}
\end{equation}

\begin{lemma}
\emph{Let $T$ be the social network matrix [cf.\ Eq.\ (\ref{decomposition}%
)]. Then, we have:
\begin{itemize}
\item[(a)] The matrix $T^k$
converges to a stochastic matrix with identical rows ${1\over n} e$
as $k$ goes to infinity, i.e.,
\[\lim_{k\to \infty} T^k = {1\over n} e e'.\]
\item[(b)] For any $z(0)\in \mathbb{R}^n$, let the sequence $z(k)$
be generated by the linear update rule
\[z(k)=T z(k-1)\qquad \hbox{for all }k\ge 0.\]
For all $k\ge 0$, define $M(k)\in \mathbb{R}$ and $m(k)\in
\mathbb{R}$ as follows:
\[M(k) = \max_{i\in {\cal N}}\, z_i(k),\qquad m(k) = \min_{i\in {\cal N}}\, z_i(k).\]
Then, for all $k\ge 0$, we have
\[M(k)-m(k)\le \delta^k(M(0)-m(0)).\]
Here $\delta>0 $ is a constant given by
\[\delta = (1-n\chi^d)^{1\over d},\]
\[\chi=\min_{(i,j)\in {\cal E}} \left\{{1\over n} \Big[p_{ij} \,
{1-\gamma_{ij}\over 2} + p_{ji} \, {1-\gamma_{ji}\over 2}
\Big]\right\},\] and $d$ is the maximum shortest path length in the
graph $({\cal N}, {\cal E})$ [cf.\ Eq.\ (\ref{maxsp})].
\end{itemize}
}\label{dslim}
\end{lemma}

\begin{proof} (a) \  By Assumption \ref{comprob}, we have for all
$i$, $\sum_{j\ne i} p_{ij}=1$ and $p_{ij}\ge 0$ for all $j$. This
implies that $\sum_{j\ne i} p_{ji}\le n-1$ and therefore
\begin{equation}1-\frac{\sum_{j\ne i} (p_{ij}+p_{ji})}{n}\ge 0\qquad
\hbox{for all }i.\label{bound}\end{equation} Since $\sum_{j\ne i}
p_{ij}=1$ for all $i$, there exists some $j$ such that $p_{ij}>0$,
i.e., $(i,j)\in {\cal E}$. By Assumption \ref{intprob}, this implies
that $\beta_{ij}+\alpha_{ij}= 1-\gamma_{ij}>0$, showing that
$T_{ii}>0$ for all $i$. Similarly, for any
 $(i,j)\in {\cal E}$, we have $p_{ij}>0$ and therefore
 $1-\gamma_{ij}>0$, showing that $T_{ij}>0$ for all $(i,j)\in {\cal
 E}$. Using Eq.\ (\ref{bound}) in Eqs.\ (\ref{Tdiagelem}) and (\ref{Toffelem}),
 it follows that for all $i$
\[[T]_{ii}\ge T_{ij}\qquad \hbox{for all }j.\]
Thus, we can use Lemma \ref{primitive} with the identification
\begin{equation}\chi = \min_{(i,j)\in {\cal E}} \left\{ {1\over n}\, \left[p_{ij}\,
{1-\gamma_{ij}\over 2} + p_{ji}\, {1-\gamma_{ji}\over
2}\right]\right\},\label{asu}\end{equation} and obtain
\begin{equation}[T^d]_{ij}\ge \chi^d\qquad \hbox{for all }i,j,\label{Tprim}\end{equation} i.e., $T$ is a
primitive matrix and therefore the Markov Chain with transition
probability matrix $T$ is regular. It follows from Theorem
\ref{MC}(a) that for any $z(0)\in \mathbb{R}^n$, we have
\[\lim_{k\to \infty} T^k z(0) = e\bar{z},\]
where $\bar{z}$ is given by $\bar{z}=\pi'z(0)$ for some probability
vector $\pi$. Since $T$ is a stochastic and symmetric matrix, it is
doubly stochastic. Denoting $z(k)=T^kz(0)$, this implies that the
average of the entries of the vector $z(k)$ is the same for all $k$,
i.e.,
\[{1\over n} \sum_{i=1}^n z_i(k) ={1\over n} \sum_{i=1}^n z_i(0)\qquad \hbox{for all }k\ge 0.\]
Combining the preceding two relations, we obtain
\[\lim_{k\to \infty} {1\over n} \sum_{i=1}^n z_i(k)= \bar{z} ={1\over n} \sum_{i=1}^n z_i(0),\]
establishing the desired relation.

\noindent (b) \ In view of Eq.\ (\ref{Tprim}), we can use Lemma
\ref{lyapdec} with the identifications
\[H(k)=T, \quad B=d,\quad \theta = \chi^d,\]
where $\chi$ is defined in Eq.\ (\ref{asu}), and obtain
\[M(k)-m(k)\le (1-n\chi^d)^{k\over d}(M(0)-m(0)).\]\end{proof}

\vskip 1pc

\noindent \textbf{\large Appendix D} \vskip .5pc \noindent \textbf{{\large %
Characterization of the Mean Commute Time, Section \ref{local-misinfo}}}

\vskip .5pc

First, we characterize the mean commute time between two nodes for a random
walk on an undirected graph using Dirichlet principle and its dual,
Thompson's principle.

\begin{definition}
\label{dirichlet_form} Consider a random walk on a weighted undirected graph
$(\mathcal{N}, \mathcal{A})$ with weight $w_{ij}$ associated to each edge $%
\{i,j\}$. Define the \emph{Dirichlet form} $\mathcal{E}$, as follows. For
functions $g:\mathcal{N }\rightarrow \mathbb{R }$ write
\begin{equation*}
\mathcal{E}(g,g) = {\frac{1}{2}} \sum_{i,j}\frac{w_{ij}}{w} \Big(g(i) - g(j)%
\Big)^2,
\end{equation*}
where $w = \sum_{i,j}w_{ij}$ is the total edge weight.
\end{definition}

\begin{lemma}
\label{extrmal_commute} \emph{Consider a random walk on a weighted undirected graph with weight
$w_{ij}$ associated to each edge $\{i,j\}$. For mean commute time between distinct nodes $a$ and
$b$ we have,
\begin{eqnarray}  \label{Dirichlet}
m_{ab} + m_{ba} &=& \sup \bigg\{ \frac{1}{\mathcal{E}(g,g)}: 0 \leq g\leq 1,
g(a) =0, g(b) = 1 \bigg\} \\
&=& w \inf \bigg\{ \frac{1}{2} \sum_{i,j} \frac{f_{ij}^2}{w_{ij}}: f \text{
is a unit flow from $a$ to $b$} \bigg\},  \label{Thompson}
\end{eqnarray}
where $m_{ab}$ is the mean first passage time from $a$ to $b$, and $w$ is
the total edge weight.}
\end{lemma}

\begin{proof}
See Section 7.2 of \cite{aldous}.
\end{proof}

It is worth mentioning that the two forms of the mean commute time
characterization in Lemma \ref{extrmal_commute} are dual of each other. The
first form is a corollary of Dirichlet principle, while the second is
immediate result of Thompson's principle. Using the electric circuit
analogy, we can think of function $g(i)$ as potential associated to node $i$%
, and flow $f_{ij}$ as the current on edge $\{i,j\}$ with resistance $\frac{1}{%
w_{ij}}$. The expressions in (\ref{Thompson}) are equivalent descriptions of
minimum energy dissipation in such electric network. Hence, we can interpret
the mean commute time between two particular nodes as the effective
resistance between such nodes in a resistive network. This allows us to use
Monotonicity Law to obtain simpler bounds for mean commute time.

\begin{lemma}
\label{monotonicity_lemma} \emph{\textbf{(Monotonicity Law)} Let $\tilde
w_{ij} \leq w_{ij}$ be the edge-weights for two undirected graphs. Then,
\begin{equation}  \label{mon_law}
m_{av} + m_{va} \leq \Big(\frac{w}{\tilde w}\Big) (\tilde m_{av} + \tilde m_{va}), \quad \text{for
all } a,v, \nonumber
\end{equation}
where $w = \sum_{i,j}w_{ij}$ and $\tilde w = \sum_{i,j}\tilde w_{ij}$ are
the total edge weight.}
\end{lemma}

\begin{proof}
Let $f^*$ and $\tilde f^*$ be the optimal solutions of (\ref{Thompson}) for the original and
modified graphs, respectively. We can write
\begin{eqnarray}
   m_{av} + m_{va} &=& \frac{w}{2} \sum_{i,j} \frac{(f^*_{ij})^2}{w_{ij}} \leq \frac{w}{2} \sum_{i,j} \frac{(\tilde
   f^*_{ij})^2}{w_{ij}} \nonumber \\
   &\leq& \frac{w}{2} \sum_{i,j} \frac{(\tilde f^*_{ij})^2}{\tilde w_{ij}} = \Big(\frac{w}{\tilde w}\Big)\frac{\tilde w}{2} \sum_{i,j} \frac{(\tilde f^*_{ij})^2}{\tilde
   w_{ij}} \nonumber \\
   &=& \Big(\frac{w}{\tilde w}\Big) (\tilde m_{av} + \tilde m_{va}), \nonumber
\end{eqnarray}
where the first inequality follows from optimality of $f^*$, and feasibility of $\tilde f^*$.
\end{proof}

By the electric network analogy, Lemma \ref{monotonicity_lemma} states that
increasing resistances in a circuit increases the effective resistance
between any two nodes in the network. Monotonicity law can be extremely
useful in providing simple bounds for mean commute times.

\newpage

\bibliographystyle{amsplain}
\bibliography{influence}

\providecommand{\bysame}{\leavevmode\hbox to3em{\hrulefill}\thinspace}
\providecommand{\MR}{\relax\ifhmode\unskip\space\fi MR }
\providecommand{\MRhref}[2]{%
  \href{http://www.ams.org/mathscinet-getitem?mr=#1}{#2}
}
\providecommand{\href}[2]{#2}
\begin{thebibliography}{10}

\bibitem{ilan}
D.~Acemoglu, Munther Dahleh, Ilan Lobel, and A.~Ozdaglar, \emph{Bayesian
  learning in social networks}, Preprint, 2008.

\bibitem{aldous}
D.~Aldous and J.~Fill, \emph{Reversible {M}arkov chains and random walks on
  graphs}, Monograph, http://www.stat.berkeley.edu/~aldous/RWG/book.html, 2002.

\bibitem{ambrus}
A.~Ambrus and S.~Takahashi, \emph{Multi-sender cheap talk with restricted state
  spaces}, Theoretical Economics \textbf{3} (2008), no.~1, 1--27.

\bibitem{bala-goyal}
V.~Bala and S.~Goyal, \emph{Learning from neighbours}, Review of Economic
  Studies \textbf{65} (1998), no.~3, 595--621.

\bibitem{bala-goyal2}
\bysame, \emph{Conformism and diversity under social learning}, Economic Theory
  \textbf{17} (2001), 101--120.

\bibitem{abhijit}
A.~Banerjee, \emph{A simple model of herd behavior}, Quarterly Journal of
  Economics \textbf{107} (1992), 797--817.

\bibitem{ban-fud}
A.~Banerjee and D.~Fudenberg, \emph{Word-of-mouth learning}, Games and Economic
  Behavior \textbf{46} (2004), 1--22.

\bibitem{BHW}
S.~Bikchandani, D.~Hirshleifer, and I.~Welch, \emph{A theory of fads, fashion,
  custom, and cultural change as information cascades}, Journal of Political
  Economy \textbf{100} (1992), 992--1026.

\bibitem{boyd}
S.~Boyd, A.~Ghosh, B.~Prabhakar, and D.~Shah, \emph{Gossip algorithms: Design,
  analysis, and applications}, Proceedings of IEEE INFOCOM, 2005.

\bibitem{bremaud}
P.~Bremaud, \emph{Markov chains: {G}ibbs fields, {M}onte {C}arlo simulation,
  and queues}, Springer, New York, 1999.

\bibitem{celen-kariv2}
B.~Celen and S.~Kariv, \emph{Distinguishing informational cascades from herd
  behavior in the laboratory}, The American Economic Review \textbf{94} (2004),
  no.~3, 484--498.

\bibitem{celen-kariv}
\bysame, \emph{Observational learning under imperfect information}, Games and
  Economic Behavior \textbf{47} (2004), no.~1, 72--86.

\bibitem{fanchung}
F.R.K. Chung, \emph{Spectral grpah theory}, American Mathematical Society,
  Providence, Rhode Island, 1997.

\bibitem{crawford-sobel}
V.P. Crawford and J.~Sobel, \emph{Strategic information transmission},
  Econometrica \textbf{50} (1982), no.~6, 1431--1451.

\bibitem{degroot}
M.H. DeGroot, \emph{Reaching a consensus}, Journal of the American Statistical
  Association \textbf{69} (1974), no.~345, 118--121.

\bibitem{demarzo}
P.M. DeMarzo, D.~Vayanos, and J.~Zwiebel, \emph{Persuasion bias, social
  influence, and unidimensional opinions}, The Quarterly Journal of Economics
  \textbf{118} (2003), no.~3, 909--968.

\bibitem{farrell}
J.~Farrell and R.~Gibbons, \emph{Cheap talk with two audiences}, American
  Economic Review \textbf{79} (1989), 1214--1223.

\bibitem{gale-kariv}
D.~Gale and S.~Kariv, \emph{Bayesian learning in social networks}, Games and
  Economic Behavior \textbf{45} (2003), no.~2, 329--346.

\bibitem{galeotti}
A.~Galeotti, C.~Ghiglino, and F.~Squintani, \emph{Strategic information
  transmission in networks}, Preprint, 2009.

\bibitem{golub-two}
B.~Golub and M.O. Jackson, \emph{How homophily affects diffusion and learning
  in networks}, Preprint, 2008.

\bibitem{golub}
\bysame, \emph{Naive learning in social networks: Convergence, influence, and
  the wisdom of crowds}, forthcoming in American Economic Journal:
  Microeconomics, 2008.

\bibitem{hagenbach}
J.~Hagenbach and F.~Koessler, \emph{Strategic communication networks},
  Preprint, 2009.

\bibitem{haviv}
M.~Haviv and L.~Van~Der Heyden, \emph{Perturbation bounds for the stationary
  probabilities of a finite {M}arkov chain}, Advances in Applied Probability
  \textbf{16} (1984), no.~4, 804--818.

\bibitem{Jacksonbook}
M.O. Jackson, \emph{Social and economic networks}, Princeton University Press,
  Princeton, New Jersey, 2008.

\bibitem{ali}
A.~Jadbabaie, J.~Lin, and S.~Morse, \emph{Coordination of groups of mobile
  autonomous agents using nearest neighbor rules}, IEEE Transactions on
  Automatic Control \textbf{48} (2003), no.~6, 988--1001.

\bibitem{kemeny}
J.G. Kemeny and J.L. Snell, \emph{Finite {M}arkov chains}, Van Nostrand, New
  York, NY, 1960.

\bibitem{MPS}
M.~Mihail, C.~Papadimitriou, and A.~Saberi, \emph{Internet is and expander},
  Proceedings of {IEEE} {S}ymposium on {F}oundations of {C}omputer {S}cience
  ({FOCS}), 2003.

\bibitem{distpaper}
A.~Nedi\'c and A.~Ozdaglar, \emph{Distributed subgradient methods for
  multi-agent optimization}, forthcoming in IEEE Transactions on Automatic
  Control, 2008.

\bibitem{reza}
R.~Olfati-Saber and R.M. Murray, \emph{Consensus problems in networks of agents
  with switching topology and time-delays}, IEEE Transactions on Automatic
  Control \textbf{49} (2004), no.~9, 1520--1533.

\bibitem{alexlong}
A.~Olshevsky and J.N. Tsitsiklis, \emph{Convergence speed in distributed
  consensus and averaging}, forthcoming in SIAM Journal on Control and
  Optimization, 2008.

\bibitem{rudin}
W.~Rudin, \emph{Real and complex analysis}, McGraw-Hill, New York,NY, 1987.

\bibitem{schweitzer}
P.J. Schweitzer, \emph{Perturbation theory and finite {M}arkov chains}, J.
  Applied Probability \textbf{5} (1968), no.~2, 401--413.

\bibitem{seneta}
E.~Seneta, \emph{Nonnegative matrices and {M}arkov chains}, Translated from
  Russian by K.C. Kiwiel and A. Ruszczynski, Springer, Berlin, 1985.

\bibitem{shimalik}
J.~Shi and J.~Malik, \emph{Normalized cuts and image segmentation}, IEEE
  Transactions on Pattern Analysis and Machine Intelligence (PAMI) \textbf{22}
  (2000), no.~8, 888--905.

\bibitem{SSnew}
L.~Smith and P.~Sorensen, \emph{Rational social learning with random sampling},
  unpublished manuscript, 1998.

\bibitem{SS}
\bysame, \emph{Pathological outcomes of observational learning}, Econometrica
  \textbf{68} (2000), no.~2, 371--398.

\bibitem{sobel}
J.~Sobel, \emph{Encyclopedia of complexity and system science}, ch.~Signaling
  Games, Springer, 2009.

\bibitem{johnthes}
J.N. Tsitsiklis, \emph{Problems in decentralized decision making and
  computation}, Ph.D. thesis, Dept. of {E}lectrical {E}ngineering and
  {C}omputer {S}cience, {M}assachusetts {I}nstitute of {T}echnology, 1984.

\bibitem{distasyn}
J.N. Tsitsiklis, D.P. Bertsekas, and M.~Athans, \emph{Distributed asynchronous
  deterministic and stochastic gradient optimization algorithms}, IEEE
  Transactions on Automatic Control \textbf{31} (1986), no.~9, 803--812.

\bibitem{watts}
D.~Watts, \emph{Six degrees: The science of a connected age}, W.W. Norton and
  Company, 2003.

\end{thebibliography}

\end{document}